\documentclass[10pt,reqno]{amsart}
\usepackage{amssymb,mathrsfs,color}
\usepackage{pinlabel}

\usepackage{amsmath, amsfonts, amsthm, verbatim, amssymb}
\usepackage{epstopdf, mathtools}
\mathtoolsset{showonlyrefs}
\usepackage{color}
\usepackage{esint}
\usepackage{stmaryrd}
\usepackage{graphicx}

\usepackage{empheq}
\usepackage[shortlabels]{enumitem}

\newcommand{\rar}{\rightarrow}

\newcommand{\cl}{\mathcal}

\def\grad {{\nabla}}

\newcommand{\bs}[1]{\boldsymbol{#1}}
\definecolor{deepgreen}{cmyk}{1,0,1,0.5}

\newcommand{\R}{\mathbb{R}}

\newcommand{\al}{\alpha}
\newcommand{\be}{\beta}

\newcommand{\de}{\delta}

\newcommand{\la}{\lambda}

\newcommand{\p}{\partial}

\makeatletter

\newcommand{\Rmnum}[1]{\expandafter\@slowromancap\romannumeral #1@}
\makeatother

\newcommand{\Del}[1]{}

\usepackage[multiple]{footmisc}
\numberwithin{equation}{section}

\newtheorem{thm}{Theorem}[section]
\newtheorem{cor}[thm]{Corollary}
\newtheorem{lem}[thm]{Lemma}
\newtheorem{prop}[thm]{Proposition}

\theoremstyle{remark}

\newtheorem{defn}[thm]{Definition}

\usepackage{tensor}

\usepackage{graphicx}
\usepackage{xcolor} % A package to add color.
\usepackage{tensor}
\usepackage{slashed}
\usepackage{cite}
\definecolor{green}{rgb}{0,0.8,0} % Redefines the color green.
\newcommand{\sgn}{{\mathrm{sgn}}}
\newcommand{\tr}{\textrm{tr}}

\newcommand{\eps}{\epsilon}
\newcommand{\veps}{\varepsilon}

\newcommand{\bbE}{\mathbb E}

\newcommand{\bbR}{\mathbb R}

\newcommand{\tens}{\otimes}

\newcommand{\Div}{\mbox{Div}\,}

\newcommand{\msm}{\bs{\mathsf{m}}}

\newcommand{\msa}{\bs{\mathsf{a}}}
\newcommand{\msb}{\bs{\mathsf{b}}}

\newcommand{\msg}{\bs{\mathsf{g}}}

\newcommand{\msu}{\bs{\mathsf{u}}}
\newcommand{\msM}{\bs{\mathsf{M}}}

\newcommand{\msF}{\bs{\mathsf{F}}}
\newcommand{\msG}{\bs{\mathsf{G}}}

\newcommand{\msB}{\bs{\mathsf{B}}}
\newcommand{\msK}{\bs{\mathsf{K}}}
\newcommand{\msL}{\bs{\mathsf{L}}}

\newcommand{\msE}{\bs{\mathsf{E}}}
\newcommand{\msC}{\bs{\mathsf{C}}}
\newcommand{\msT}{\bs{\mathsf{T}}}
\newcommand{\mst}{\bs{\mathsf{t}}}
\newcommand{\msH}{\bs{\mathsf{H}}}

\newcommand{\msP}{\bs{\mathsf{P}}}
\newcommand{\msR}{\bs{\mathsf{R}}}
\newcommand{\msp}{\bs{\mathsf{p}}}
\newcommand{\msU}{\bs{\mathsf{U}}}
\newcommand{\sK}{\mathsf{K}}
\newcommand{\sL}{\mathsf{L}}
\newcommand{\sanl}{\mathsf{l}}
\newcommand{\sE}{\mathsf{E}}
\newcommand{\sB}{\mathsf{B}}
\newcommand{\sanb}{\mathsf{b}}
\newcommand{\sG}{\mathsf{G}}
\newcommand{\sg}{\mathsf{g}}
\newcommand{\sa}{\mathsf{a}}
\newcommand{\sh}{\mathsf{h}}
\newcommand{\sk}{\mathsf{k}}
\newcommand{\st}{\mathsf{t}}
\newcommand{\su}{\mathsf{u}}
\newcommand{\sT}{\mathsf{T}}
\newcommand{\sU}{\mathsf{U}}
\begin{document}

\title[Strain-gradient elastic surfaces]{Elastic solids with strain-gradient elastic boundary surfaces}

\author{C. Rodriguez}

\begin{abstract}
Recent works have shown that in contrast to classical linear elastic fracture mechanics, endowing crack fronts in a brittle Green-elastic solid with Steigmann-Ogden surface elasticity yields a model that predicts bounded stresses and strains at the crack tips for plane-strain problems. However, singularities persist for anti-plane shear (mode-III fracture) under far-field loading, even when Steigmann-Ogden surface elasticity is incorporated. 

This work is motivated by obtaining a model of brittle fracture capable of predicting bounded stresses and strains for all modes of loading. We formulate an exact general theory of a three-dimensional solid containing a boundary surface with strain-gradient surface elasticity. For planar reference surfaces parameterized by flat coordinates, the form of surface elasticity reduces to that introduced by Hilgers and Pipkin, and when the surface energy is independent of the surface covariant derivative of the stretching, the theory reduces to that of Steigmann and Ogden.  We discuss material symmetry using Murdoch and Cohen's extension of Noll's theory. We present a model small-strain surface energy that incorporates resistance to geodesic distortion, satisfies strong ellipticity, and requires the same material constants found in the Steigmann-Ogden theory. 

Finally, we derive and apply the linearized theory to mode-III fracture in an infinite plate under far-field loading. We prove that there always exists a unique classical solution to the governing integro-differential equation, and in contrast to using Steigmann-Ogden surface elasticity, our model is consistent with the linearization assumption in predicting finite stresses and strains at the crack tips.  
\end{abstract}

\maketitle

\section{Introduction}

\subsection{Surface stressed solid bodies}
The study of surface tension for solids was initiated by Gibbs in 1857, and it is now widely accepted that surface tension and more general surfaces stresses must be accounted for when modeling mechanical structures at small length scales. In particular, interfaces between a three-dimensional solid and its environment can form due to various mechanisms including coating or atomic rearrangement during fracture.\footnote{By an interface, we mean a thin region separating either two distinct materials or two distinct phases of a material.} A way to mathematically model such an interface is by endowing part of the three-dimensional body's two-dimensional boundary with it's own thermodynamic properties that are distinct from the substrate (such as energy). 

Motivated by modeling the observed compressive surface stresses in certain cleaved crystals, Gurtin and Murdoch \cite{GurtinMurd75, GurtinMurd78} developed a rigorous general theory of material surface stresses that accounts for the surface's resistance to stretching (but not flexure). Their celebrated theory has been used to model a wide arrange of phenomena over the past 45 years, especially recently due to advances in nanoscience and nanotechnology. In the special case of a Green-elastic, three-dimensional solid with reference configuration $\cl B$ and material surface $\cl S \subseteq \p \cl B$, the field equations for the current configuration $\bs \chi: \cl B \rar \bs \chi(\cl B)$ are the Euler-Lagrange equations associated to the total strain energy
\begin{align}
	\bs \Phi[\bs \chi] = \int_{\cl B} W(\bs C) \, dA + \int_{\cl S} U(\msC) \, dA,
\end{align} 
where we omit the possible explicit dependence of the functions $W$ and $U$ on points in $\cl B$ and $\cl S$ (see Figure \ref{f:1}).
Here $\bs C$ is the left Cauchy-Green stretch tensor and $\msC$ is the left Cauchy-Green surface stretch tensor, the pullbacks by $\bs \chi$ of the metric tensors on $\bs \chi(\cl B)$ and $\bs \chi(\cl S)$, respectively. In particular, the material surface stress tensor is derived from the surface energy density $U(\msC)$ in analogy with the substrate's Piola stress being derived from the substrate energy density $W(\bs C)$. 

\begin{figure}[b]
	\centering
	\includegraphics[scale=.5]{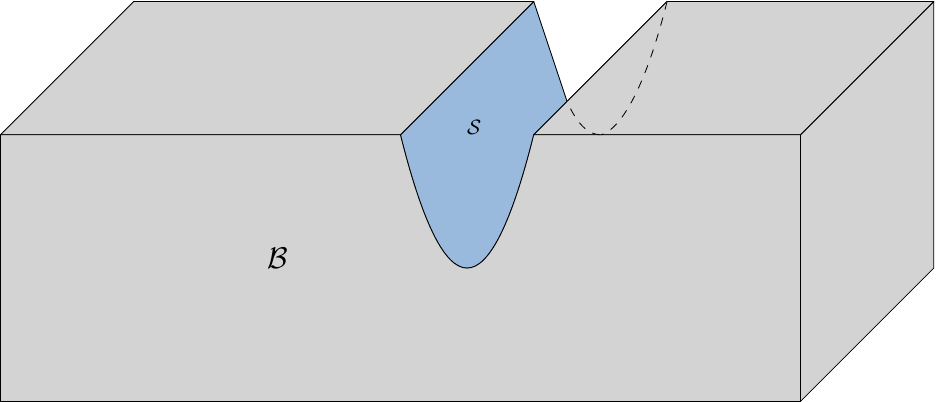}
	\caption{The reference configuration of a three-dimensional solid $\cl B$ with strain-gradient elastic surface $\cl S \subseteq \p \cl B$.}
	\label{f:1}
\end{figure}

However, Steigmann and Ogden showed in their seminal works \cite{SteigOgden97a, SteigOgden99} that { surface-substrate} equilibrium states under compressive surface stresses obtained from the Gurtin-Murdoch theory do not satisfy an associated Legendre-Hadamard condition, and thus, these equilibrium states cannot be local energy minimizers.\footnote{ The fact that a stand-alone membrane (with no substrate), a special type of Gurtin-Murdoch material surface, in equilbrium and under compressive surface stresses cannot by locally energy minimizing was proved in two independent works: by Pipkin \cite{Pipkin86} for initially plane, isotropic membranes, and by Steigmann \cite{Steigmann86} for general, possibly anistropic and initially curved membranes.} Steigmann and Ogden \cite{SteigOgden97a, SteigOgden99} rectified this inadequacy and incorporated the material surface's resistance to flexure by including normal curvature dependence in the surface energy: 
\begin{align}
	\bs \Phi[\bs \chi] = \int_{\cl B} W(\bs C) \, dA + \int_{\cl S} U(\msC, \bs \kappa) \, dA,
\end{align}
where $\bs \kappa$ is the pullback of the second fundamental form on $\bs \chi(\cl S)$. Since $\bs \kappa$ depends on the second derivatives of $\bs \chi$, the surface energy is of \emph{strain-gradient} type. In recent years, the Steigmann-Ogden theory has attracted considerable interest from various perspectives including the study of contact problems \cite{Zem18, MI2018, Zem19, LiMi19a, LiMi19b, LiMi21, ZemWhite22} and inclusion problems \cite{HAN2018, ZEMLYANOVA2018, ZEMMOG18B, DAI2018, MogKZem19, WANG2020311}. The theory has also been used to model fracture in brittle materials \cite{Zem17StraightCrack, Zem21Penny, Zem20Multiple}, the main phenomenon motivating this work. 

One of the most successful and practical theories for modeling fracture in brittle materials is classical linear elastic fracture mechanics. The governing linear partial differential equations are derived from finite elasticity under the assumption of \emph{infinitesimal} strains (linearized elasticity), but the theory predicts \emph{unbounded} singular strains at the crack tips, a striking inconsistency and physically impossible prediction. There have been a vast number of suggestions for supplementing classical linear elastic fracture mechanics to correct this defect in the theory (see, e.g., \cite{Broberg}). 

{ One modification of linear elastic fracture mechanics that removes certain singularities is including higher spatial gradients of the displacement in the stored energy. In the foundational works of Muki and Sternberg \cite{SternbergMuki65, SternbergMuki67} and Bogy and Sternberg \cite{BogySternberg67, BogySternberg68}, the authors found that including couple-stresses removed the singularities in the rotation gradient predicted by classical linear elastostatics in a variety of problems. However, crack tip stress and strain singularities of the same order but with different intensity persisted for mode-I fracture \cite{SternbergMuki67}. Including strain-gradients rather than the rotation gradient in the body's stored energy removes crack tip strain singularities, but stress singularities still persist in general (see, e.g., \cite{AltanAif1992, Askes2011}).}

A more recent approach aimed at eliminating \textit{both stress and strain} crack tip singularities is to modify the boundary conditions of classical linear elastic fracture mechanics by endowing the crack fronts with material surface stresses. { This methodology does not introduce higher gradients in the substrate's stored energy, but it does include higher surface gradients via surface elasticity.} Beginning with \cite{OhWaltonSlatt} and developed further by Sendova and Walton \cite{SendovaWalton}, one line of thought has been to prescribe the crack front's surface stresses as a normal curvature dependent surface tension. Although able to remove the singularities completely in a diverse range of settings (see \cite{ZemWalton12, Zem12Curvature, Zem14, Walton2014plane, Ferguson2015numerical}), it is unclear if the surface stresses from the Sendova-Walton theory can be derived from a surface energy density, a reasonable definition of ``elastic-like" behavior.\footnote{Equivalently, it is unclear if the governing field equations from the Sendova-Walton model can be derived from a Lagrangian energy functional.} 

{Deriving material stresses for the crack fronts from a Gurtin-Murdoch surface energy \emph{does not} remove the crack tip singularities for far-field loading \cite{WaltonNote12, IMRUSch13}. {The refined analysis given in \cite{Gorbushinetal20} shows that different, nonuniform, mode-III loading along crack fronts possessing Gurtin-Murdoch surface elasticity with inhomogeneous surface shear constants may either lead to a logarithmic singularities, classical square root singularities, or bounded stresses and strains at the crack tips.} However, endowing the crack fronts with Steigmann-Ogden surface energy \emph{does} remove the singularities for certain plane-strain problems \cite{Zem17StraightCrack, Zem20Multiple} and axisymmetric penny shaped cracks \cite{Zem21Penny}. Unfortunately, for anti-plane shear (mode-III loading), Steigmann-Ogden surface energy reduces to Gurtin-Murdoch surface energy, and \emph{crack tip singularities persist} for far-field loading. Our main motivation is to develop a model of brittle fracture capable of predicting bounded crack tip stresses and strains for all modes of loading. Inspired by \cite{Ferguson2015numerical}, this work proposes a surface-substrate equilibrium theory in which the surface energy depends on stretching, normal curvature, and the surface covariant derivative of stretching.}\footnote{The fact that the Steigmann-Ogden theory reduces to the Gurtin-Murdoch theory for mode-III loading follows from the vanishing of the linearized normal curvature for anti-plane shear, i.e., for displacement fields tangent to the material surface.}

\subsection{Main results and outline}
This work proposes an augmentation of the Steigmann-Ogden surface-substrate theory that is of strain-gradient type and includes the surface covariant derivative of stretching in the surface energy: 
\begin{align}
	\bs \Phi[\bs \chi] = \int_{\cl B} W(\bs C) \, dA + \int_{\cl S} U(\msC, \bs \kappa, \nabla \msC) \, dA, \label{eq:introenergy}
\end{align}
where $\nabla$ is the Levi-Cevita connection on $\cl S$. For $\cl S$ contained in a plane and parameterized by flat coordinates, the form of the surface energy $U$ is equivalent to that introduced by Hilgers and Pipkin \cite{HilgPip92a} for strain-gradient elastic plates. { For general, possibly curved $\cl S$, the form of $U$ appearing in \eqref{eq:introenergy} is equivalent to that proposed by Steigmann \cite{Steigmann18Lattice} for stand-alone strain-gradient elastic surfaces with no substrate.} As we show in Section 2, the dependence of the surface energy on $\nabla \msC$ endows the material surface $\cl S$ with resistance to \emph{geodesic distortion}, i.e., convected geodesics failing to be geodesics on $\bs \chi(\cl S)$. 

{ Incorporating the resistance of an elastic surface to geodesic distortion is not solely a mathematical exercise. In addition to the modeling of fracture presented in Section 4 of this work, resistance to geodesic distortion has also proven to be an essential tool in the modeling of \emph{fibrous surfaces}.
	
A model for a three-dimensional solid with a boundary surface formed by a system of aligned thin elastic fibers (\emph{hyperbolic metasurfaces}) was introduced by Eremeyev \cite{Eremeyev19}, and the equilibrium theory serves as a special case of the general theory presented in this work. There, a single unit vector field $\bs T$ on the reference surface $\cl S \subseteq \cl B$ is interpreted as being tangent to a continuous distribution of aligned fibers. Assuming that the reference fibers are geodesics, the surface energy $U$ depends on $\msC$, the stretch of the fibers, and the absolute value of the geodesic curvature of the fibers.
	
The modeling of stand-alone, two-dimensional, fibrous \emph{networks} that are not attached to substrates has also been a source of intensive study. In this theory, two orthonormal vector fields $\bs L$ and $\bs M$ on $\cl S$ are interpreted as being tangent to continuous distributions of fibers forming a reference network. The current network corresponds to the integral curves of the convected vector fields. The surface energy $U$ yielding the equilibrium theory via Hamilton's principle is a function of kinematic quantities associated to the curves forming the current network (e.g., their stretches, normal curvatures, tangent curvatures, etc.). The foundational continuum theories by Rivlin \cite{Rivlin55}, Green and Adkins \cite{GreenAdkins1970Book}, Pipkin \cite{Pipkin81}, and Steigmann-Pipkin \cite{SteigmannPipkin91} focused on networks of perfectly flexible fibers with surface energies depending on the stretches and shear angle.  
	
Later, researchers incorporated second gradient effects in the theory of fibrous networks, resulting in surface energies of the form appearing in \eqref{eq:introenergy}. Wang and Pipkin \cite{WangPipkin86} proposed a theory of networks of inextensible fibers with flexural resistance, and more recently, Steigmann \cite{Steigmann18CrossedElastb} generalized their theory to include resistance to fiber twist. The full second surface covariant derivative of the deformation and resistance to geodesic bending of the fibers were first included in surface energies by Steigmann and dell’Isola in \cite{SteigmanndellIsola15}. There, the authors decomposed the second covariant derivative of the deformation of $\cl S$ in the $\{\bs L, \bs M\}$ basis with respect to the two covariant slots. Assuming $\bs L$ and $\bs M$ generate geodesics, certain terms can be shown to be proportional to the geodesic curvatures of the convected fibers (see, e.g., (52) in \cite{SteigmanndellIsola15}).  In \cite{Steigmann18Lattice}, Steigmann developed the fully general equilibrium theory for stand-alone elastic surfaces with surface energy depending on the surface deformation gradient and second covariant derivative of the deformation (with no substrate), including material symmetry and a virtual work principle.  Numerical solutions of the Steigmann-dell'Isola theory correctly predicted regions of uniform shear separated by thin regions of geodesic bending observed in bias tests of uniform pantographic lattices \cite{Giorgioetal15}. Further numerics based on the Steigmann-dell'Isola theory predicted novel “geodesic buckling” in plane shear deformations \cite{Giorgioetal16}, bulging effects in strain-gradient elastic cylinders that are in sharp contrast to the membrane theory \cite{Giorgioetal18}, and strain energy localization at the edges of Hypar nets \cite{Giorgioetal19}. 
	
As the above works indicate and as we demonstrate in Section 4, including strain-gradients and resistance to geodesic distortion in surface elasticity presents a powerful and geometric continuum approach to accurately modeling materials with small-scale structures.  
}    

In Section 2, the relevant kinematics of the boundary surface $\cl S$ convected by a deformation of the three-dimensional solid $\cl B$ are first summarized. In particular, we introduce a third order tensor $\msL$ that is obtained from certain transposes of $\nabla \msC$ and has components in terms of the difference of Christoffel symbols on $\cl S$ and $\bs \chi(\cl S)$. This tensor is used in a model surface energy proposed in Section 3. Physically, the tensor $\msL$ furnishes the rate of stretching of convected geodesics and locally characterizes the notion of geodesic distortion previously discussed (see Proposition \ref{p:geod}).  The general form of \eqref{eq:introenergy} that we consider is then introduced (see \eqref{eq:strainenergy}). We then discuss material symmetry for the surface energy density $U$ (see \ref{eq:symmetryset}) using Murdoch and Cohen's extension to surfaces of Noll's classical theory of material symmetry that was reformulated by Steigmann and Ogden \cite{SteigOgden99} and Steigmann \cite{Steigmann18Lattice}. In the final subsection of Section 2, we derive the field equations \eqref{eq:fieldequations} governing equilibrium states of the solid $\cl B$ with strain-gradient elastic surface $\cl S \subseteq \p \cl B$ from a Lagrangian energy functional \eqref{eq:action} including the boundary relations between the boundary tractions and the surface stresses.

In Section 3, we present a properly hemitropic, small-strain surface energy that requires the same number of material constants (with the same physical interpretations) as found in the Steigmann-Ogden theory (see \eqref{eq:surfaceen}). In contradistinction, however, the surface energy we propose satisfies the strong ellipticity condition (see \eqref{eq:strongellipt}, \eqref{eq:strong}). We then derive the linearized field equations governing infinitesimal displacements of $\cl B$ and $\cl S$, including the boundary relations connecting the linearized boundary tractions and the linearized surface stresses. { The resulting linear theory is equivalent to the $N = 2$ case of the linear surface-substrate model with $N$th order surface elasticity proposed by Eremeyev, Lebedev and Cloud \cite{Eremeyevetal21}. In particular, our work may be seen as giving the exact theory that yields the $N = 2$ model from \cite{Eremeyevetal21} for infinitesimal displacements.} 

Finally, in Section 4, we apply the linearized theory \eqref{eq:HPlinfieldequations} to the problem of a brittle infinite solid with a straight, non-interfacial crack  of finite length, under mode-III loading. Using the explicit Dirichlet-to-Neumann map, the problem is reduced to solving a fourth order integro-differential equation for the crack profile along the boundary crack front (see \eqref{eq:integrodiff}).  Analytically, it is the surface energy satisfying the strong ellipticity condition that implies that this integro-differential equation is fourth order in the surface derivative of the displacement.\footnote{Physically, it is the model incorporating resistance to geodesic distortion that implies that this integro-differential equation is fourth order in the surface derivative of the displacement.} The dimensionless form of the equation implies that the behavior of the displacement depends on the size of the crack, and for macro cracks, we expect the displacement to be well-approximated by the solution given by classical linear elastic fracture mechanics except in small regions near the crack tips (boundary layers). Finally, using the Lax-Milgram theorem and regularity afforded by the presence of the fourth order derivative, we prove that there always exist a unique classical solution to the governing integro-differential equation, and in contrast to using the Steigmann-Ogden theory, our model predicts bounded stresses and strains up to the crack tips (see Theorem \ref{t:bounded}). 

\subsection*{Acknowledgments}{The completion of this work was supported by NSF Grant DMS-2307562. The author thanks Jay R. Walton for several fruitful discussions related to fracture, surface elasticity and interfacial physics. The author would also like to thank Jeremy Marzuola for his help generating numerical solutions to the mode-III fracture problem discussed in Section 4. Finally, the author thanks the anonymous referees whose careful reading greatly improved the clarity of exposition and scholarship of this work.}

\section{Kinematics and field equations}

In this section formulate a general mathematical model for a Green-elastic three-dimensional solid containing a boundary surface with strain-gradient surface elasticity. For planar reference surfaces, the form of surface elasticity reduces to that introduced by Hilgers and Pipkin \cite{HilgPip92a}, and when the surface energy is independent of the surface covariant derivative of the stretching, the theory reduces to that of Steigmann and Ogden \cite{SteigOgden97a, SteigOgden99}. 

\subsection{Kinematics}
Let $\bbE^3$ be three-dimensional Euclidean space. We identify the translation space of $\mathbb E^3$ with $\bbR^3$ via a fixed orthonormal basis $\{\bs e_i \}_{i = 1}^3$. Upon choosing an origin $\bs o \in \bbE^3$, we identify subsets of $\bbE^3$ with subsets of $\bbR^3$ via their position vectors: $\bbE^3 \ni \bs p \mapsto \bs p - \bs o \in \bbR^3$.   We define the following operations for elementary tensor products of vectors in $\bbR^3$, 
\begin{gather}
	(\bs a_1 \tens \bs a_2)\bs b = (\bs a_2 \cdot \bs b) \bs a_1, \quad 
	(\bs a_1 \tens \bs a_2 \tens \bs a_3) \bs b = (\bs a_3 \cdot \bs b) \bs a_1 \tens \bs a_2, \\
	(\bs a_1 \tens \bs a_2 \tens \bs a_3)(\bs b_1 \tens \bs b_2) = (\bs a_3 \cdot \bs b_1) \bs a_1 \tens \bs a_2 \tens \bs b_2, \\
	(\bs b_1 \tens \bs b_2)(\bs a_1 \tens \bs a_2 \tens \bs a_3) = (\bs b_2 \cdot \bs a_1) \bs b_1 \tens \bs a_2 \tens \bs a_3, \\
	(\bs a_1 \tens \bs a_2 \tens \bs a_3)[\bs b_1 \tens \bs b_2] = 
	(\bs a_2 \cdot \bs b_1)(\bs a_3 \cdot \bs b_2) \bs a_1, \\
	(\bs a_1 \tens \bs a_2)[\bs b_1 \tens \bs b_2] = (\bs a_1 \cdot \bs b_1)(\bs a_2 \cdot \bs b_2), \\
	(\bs a_1 \tens \bs a_2 \tens \bs a_3)[\bs b_1 \tens \bs b_2 \tens \bs b_3] = 
	(\bs a_1 \cdot \bs b_1)(\bs a_2 \cdot \bs b_2)(\bs a_3 \cdot \bs b_3), \\
	(\bs a_1 \tens \bs a_2)^T = \bs a_2 \tens \bs a_1, \quad
	(\bs a_1 \tens \bs a_2 \tens \bs a_3)^T = \bs a_1 \tens \bs a_3 \tens \bs a_2, \\
	(\bs a_1 \tens \bs a_2 \tens \bs a_3)^\sim = \bs a_2 \tens \bs a_1 \tens \bs a_3.  
\end{gather} 
These operations are extended to general second and third order tensors on $\bbR^3$ by linearity. 

Let $\cl B \subseteq \bbE^3$ be a domain with smooth boundary $\p \cl B$, the reference configuration of a three-dimensional body, and let $\cl S \subseteq \p \cl B$ be a closed surface with smooth boundary. Let $\bs \chi : \cl B \rar \bs \chi(\cl B) \subseteq \bbE^3$ be a smooth, invertible deformation of $\cl B$. We denote the current position of the reference particle $\bs X \in \cl B$ by 
\begin{align}
	\bs x = \bs \chi (\bs X) = (\chi^1(\bs X), \chi^2(\bs X), \chi^3(\bs X)). 
\end{align}
The deformation gradient $\bs F : \bbR^3 \rar \bbR^3$ is the second order tensor field
\begin{align}
	\bs F
	&= \frac{\p \chi^i}{\p X^a}(\bs X) \bs e_i \tens \bs e^a
\end{align}
where $\bs e^a := \bs e_a$, $a = 1, 2, 3$.
We denote the left Cauchy-Green stretch tensor by $\bs C =  \bs F^T \bs F$. 

Let $\bs Y = \hat{\bs Y}(\theta^1, \theta^2)$ be a local parameterization of the reference surface $\cl S \subseteq \p \cl B$. Then 
\begin{align}
	\bs y = \hat{\bs y}(\theta^1, \theta^2) = \bs \chi(\hat{\bs Y}(\theta^1, \theta^2))
\end{align}
is a local parameterization of the current surface $\bs \chi(\cl S) \subseteq \bs \chi(\p \cl B)$. The (local) tangent vector fields on the reference and current surfaces are then given by 
\begin{gather}
	\bs Y_{, \al} \in T_{\bs Y} \cl S, \quad 
	\bs y_{, \al} = \bs F \bs Y_{,\al} \in T_{\bs y} \bs \chi(\cl S), 
\end{gather} 
where $\cdot_{, \al} := \frac{\p}{\p \theta^\al}$. The dual tangent vector fields on the reference and current surfaces are denoted by $\bs Y^{,\beta}$ and $\bs y^{,\beta}$ respectively and satisfy   
\begin{align}
	\bs Y^{,\beta} \cdot \bs Y_{,\al} = \tensor{\delta}{^\beta_\al}, \quad \bs y^{,\beta} \cdot \bs y_{,\al} = \tensor{\delta}{^\beta_\al}.
\end{align} 
We remark that we may also write 
\begin{align}
	\bs y_{,\al} = \msF \bs Y_{,\al}
\end{align}
where $\msF = \tensor{F}{^k_{\beta}} \bs e_k \tens \bs Y^{,\beta} :=  y^k_{,\beta} \bs e_k \tens \bs Y^{,\beta} = \bs y_{,\beta} \tens \bs Y^{,\beta}$ is the \emph{surface deformation gradient}. The first fundamental forms for the reference and current surfaces are then given by 
\begin{gather}
	\msG = \sG_{\al \beta} \bs Y^{,\alpha} \tens \bs Y^{,\beta}, \quad \sG_{\al \beta} = \bs Y_{,\al} \cdot \bs Y_{,\beta}, \\
	\msg = \sg_{\al \beta}  \bs y^{,\alpha} \tens \bs y^{,\beta}, \quad \sg_{\al \beta} = \bs y_{,\alpha} \cdot \bs y_{,\beta},
\end{gather}
and we note that 
\begin{align}
	\bs Y^{,\alpha} = (\sG^{-1})^{\al \beta} \bs Y_{,\beta}, \quad \bs y^{,\alpha} = (\sg^{-1})^{\al \beta} \bs y_{,\beta}
\end{align}
The Christoffel symbols of the second kind for the reference and current surfaces are denoted by 
\begin{align}
\tensor{\Gamma}{^\al}{_{\beta \delta}} = \bs Y^{,\alpha} \cdot \bs Y_{,\beta \delta}, \quad \tensor{\gamma}{^\al}{_{\beta \delta}} = \bs y^{,\alpha} \cdot \bs y_{,\beta \delta} . 
\end{align}

The \emph{left Cauchy-Green surface stretch tensor} on $\cl S$ is the second order tensor  
\begin{align}
	\msC := \msF^T \msF =
	\sg_{\al \beta} \bs Y^{,\alpha} \tens \bs Y^{,\beta},
\end{align}
and the {\emph surface Green-St. Venant tensor} is the second order tensor $\msE = \frac{1}{2}(\msC - \msG)$. We assume that the reference and current surfaces are orientable with unit normals to the reference and current surfaces (locally) given by 
\begin{align}
	\bs N = |\bs Y_{,1} \times \bs Y_{,2}|^{-1} \bs Y_{,1} \times \bs Y_{,2}, \quad  
	\bs n = |\bs y_{,1} \times \bs y_{,2}|^{-1} \bs y_{,1} \times \bs y_{,2}
\end{align}
respectively. The second fundamental forms on the reference and current surfaces are given by 
\begin{gather}
	\msB = \sB_{\al \beta} \bs Y^{,\alpha} \tens \bs Y^{,\beta}, \quad \sB_{\al \beta} := \bs Y_{,\al\beta} \cdot \bs N, \\
	\msb = \sanb_{\al \beta} \bs y^{,\alpha} \tens \bs y^{,\beta}, \quad \sanb_{\al \beta} := \bs y_{,\al \beta} \cdot \bs n. 
\end{gather} 
The \emph{relative normal curvature tensor} on $\cl S$ is the second order tensor $\msK = \sK_{\al \beta} \bs Y^{,\al} \tens \bs Y^{,\beta}$ given by 
\begin{align}
	\msK = \msF^T \msb \msF - \msB = 
	[\sanb_{\al \beta} - \sB_{\al \beta}] \bs Y^{,\alpha} \tens \bs Y^{,\beta}. 
\end{align}
As discussed by Steigmann and Ogden \cite{SteigOgden99}, $\msE$ and $\msK$ furnish local differences in length and scaled extrinsic normal curvature between a given curve on $\cl S$ and the convected curve on $\bs \chi(\cl S)$, where $\dot{\mbox{}} := \frac{d}{ds}$. Indeed, if $\bs Z(s)$ is a a curve on $\cl S$ with unit tangent vector field $\msT$ and $\bs z(s) = \bs \chi(\bs Z(s))$ is the convected curve, then the tangent to the convected curve is given by $\mst = \dot{\bs z} = \bs F \msT = \msF \msT$. The \emph{stretch} $\nu$ of the convected curve is  
\begin{align}
	\nu = |\mst|^2 - 1 = \msF \msT \cdot \msF \msT - 1 
	= \msC \msT \cdot \msT - 1 = 2 \msE \msT \cdot \msT. 
\end{align} 
The extrinsic normal curvature of the convected curve is 
\begin{align}
	\kappa = \msb \mst \cdot \mst = |\mst|^{-2} \msb \msF \msT \cdot \msF \msT
	= (\nu+1)^{-1} (\msF^T \msb \msF) \msT, 
\end{align}
and thus, the difference in scaled extrinsic normal curvature of the convected curve and the original curve is given by 
\begin{align}
	|\mst|^2 \kappa - |\msT|^2 \msB \msT \cdot \msT = \msK \msT \cdot \msT. 
\end{align}

The Levi-Cevita connection on $\cl S$ is denoted by $\nabla$, so 
\begin{align}
	\nabla \msE = \frac{1}{2} \nabla \msC = \frac{1}{2} \nabla_{\delta} \sg_{\al \beta} \bs Y^{,\alpha} \tens \bs Y^{,\beta} \tens \bs Y^{,\delta},
\end{align}
where 
\begin{align}
	\nabla_{\delta} \sg_{\al \beta} := \sg_{\al \beta, \delta} - \tensor{\Gamma}{^\mu}{_{\alpha \delta}} \sg_{\mu \beta} - \tensor{\Gamma}{^\mu}{_{\beta \delta}} \sg_{\mu \al}. 
\end{align}
For later use, we define the following third order tensor on $\cl S$,
\begin{align}
	\msL = \nabla \msE + (\nabla \msE)^T - (\nabla \msE^\sim )^T,
\end{align}
with components 
\begin{align}
	\sL_{\al \beta \delta} = \frac{1}{2} (\nabla_\delta \sg_{\al \beta} + \nabla_\beta \sg_{\al \delta} - \nabla_\al \sg_{\beta \delta}) =  (\tensor{\gamma}{^\mu_\beta_\delta} - \tensor{\Gamma}{^\mu_\beta_\delta}) \sg_{\mu \al}. 
\end{align} 
Physically, the tensor $\msL$ (and thus $\nabla \msE$) furnishes the rate of stretching of convected geodesics, and more generally, it quantifies \emph{geodesic distortion}, i.e. how convected geodesics fail to be geodesics on $\bs \chi(\cl S)$.\footnote{ We comment that the tensor $\bs \sL$ in this work and the tensor $\bs{\mathsf{S}} := \tensor{S}{^\nu_\al_\beta} \bs Y_{,\nu} \tens \bs Y^{,\al} \tens \bs Y^{,\beta}$ that plays a prominent role in \cite{SteigmanndellIsola15, Giorgioetal15, Giorgioetal16, Giorgioetal18, Giorgioetal19, Steigmann18Lattice} are related via $\sL_{\al \beta \delta} = \tensor{S}{^\nu_\al_\beta} \sg_{\nu \delta}$.} More precisely, we have the following. 
\begin{prop}\label{p:geod}
	Let $\bs Z(s) : I \rar \cl S$ be a geodesic on $\cl S$ with unit tangent vector field $\msT = \dot{\bs Z}$.  Then $\msL$ yields the rate of stretching of the convected curve $\bs z(s) = \bs \chi(\bs Z(s))$, 
	\begin{align}
		 2\msL[\msT \tens \msT \tens \msT] = \frac{d}{ds} |\dot{\bs z}|^2. \label{eq:ratestretching}
	\end{align}
	Moreover, the convected curve $\bs z(\cdot)$ is a geodesic on $\bs \chi(\cl S)$ if and only if { along the geodesic $\bs Z(\cdot)$, we have} 
	\begin{align}
	{ \msL[\msT \tens \msT] = \bs 0.} \label{eq:Lzero}
	\end{align}
\end{prop}

\begin{proof}
To prove \eqref{eq:ratestretching}, we simply note that since $\bs Z$ is a geodesic, $\nabla_{\msT} \msT = \bs 0$ so 
\begin{align}
	\frac{d}{ds} |\dot{\bs z}|^2 &= \frac{d}{ds} [\msC \msT \cdot \msT]\\ 
	&= \nabla \msC[\msT \tens \msT \tens \msT] + 2 \msC \nabla_{\msT} \msT \cdot \msT \\
	&= \nabla \msC[\msT \tens \msT \tens \msT] = 2 \msL[\msT \tens \msT \tens \msT].
\end{align}	
	
We now prove \eqref{eq:Lzero}. The convected curve $\bs z$ on $\bs \chi(\cl S)$ has tangent vector field $\mst = \msF \msT = \st^\mu \bs y_{,\mu}$ and acceleration vector field 
\begin{align}
	\msa = \bigl ( \dot \st^\mu + \tensor{\gamma}{^\mu_\beta_\delta} \st^\beta \st^\delta \bigr ) \bs y_{,\mu}. 
\end{align}
The acceleration is zero and $\bs z$ is a geodesic if and only if for each $s_0 \in I$, 
\begin{align}
		\forall \msu \in T_{\bs z(s_0)} \bs \chi(\cl S), \quad	\msu \cdot \msa |_{s = s_0} = 0. \label{eq:zeroacc}
\end{align}
We now show that \eqref{eq:zeroacc} is equivalent to \eqref{eq:Lzero}. We assume without loss of generality that $s_0 = 0$, and we choose normal coordinates $(\theta^1, \theta^2)$ centered at $\bs Z(0)$.  Then for all $\al, \beta, \delta = 1,2$, 
\begin{align}
	\sG_{\al \beta} |_{s = 0} = \delta_{\al \beta}, \quad \tensor{\Gamma}{^\mu_\beta_\delta}|_{s = 0} = 0, 
\end{align}
and there exist $\sT^1, \sT^2 \in \bbR$ with $\delta_{\al \beta} \sT^\al \sT^\beta = 1$ such that $\bs Z(s) = \hat{\bs Y}(s \sT^1, s \sT^2)$. In particular, we conclude that 
\begin{align}
 \msa = \tensor{\gamma}{^\mu_\beta_\delta} \sT^\beta \sT^\delta \bs y_{,\mu}.
\end{align}
Since $\msF |_{\bs Z(0)} : T_{\bs Z(0)} \cl S \rar T_{\bs z(0)} \bs \chi(\cl S)$ is an isomorphism, \eqref{eq:zeroacc} is equivalent to 
\begin{gather}
			\forall \msU = \sU^\mu \bs Y_{,\mu} |_{\bs Z(0)} \in T_{\bs Z(0)} \cl S, \quad	\msF \msU \cdot \msa \big |_{s = 0} = 0 \\
			\iff \forall \msU = \sU^\mu \bs Y_{,\mu} |_{\bs Z(0)} \in T_{\bs Z(0)} \cl S, \quad
			\sg_{\mu \al} \sU^\al \tensor{\gamma}{^\mu_\beta_\delta} \sT^\beta \sT^\delta \big |_{s = 0} = 0 \\
			\iff \forall \msU = \sU^\mu \bs Y_{,\mu} |_{\bs Z(0)} \in T_{\bs Z(0)} \cl S, \quad \msL[\msU \tens \msT \tens \msT] \big |_{s = 0} = 0. 
\end{gather}
\end{proof}

As an illustrative example, consider $$\cl S = \{ (X^1, X^2, 0) \mid X^1 \in [a,b], X^2 \in [0, \pi] \} \subseteq \cl B = [a,b] \times [0,\pi] \times [0,\infty),$$  
and the deformation 
\begin{align}
	\bs \chi(X^1, X^2, X^3) = (e^{X^1}\cos X^2, e^{X^1} \sin X^2, X^3).   
\end{align}
Then 
\begin{gather}
	\msE = \frac{1}{2}(e^{2X^1}-1)\bigl ( \bs e_1 \tens \bs e_1 + \bs e_2 \tens \bs e_2 \bigr ), \quad \msK = \bs 0, \\
	\msL = e^{2X^1} \bigl ( \bs e_1 \tens \bs e_1 \tens \bs e_1 + \bs e_2 \tens \bs e_1 \tens \bs e_2 + \bs e_2 \tens \bs e_2 \tens \bs e_1 - \bs e_1 \tens \bs e_2 \tens \bs e_2 \bigr ). 
\end{gather}
The image of the coordinate curve $X^2 = d$ is a straight, radially outward traveling curve in the $x^1x^2$-plane parameterized by $X^1 \in [a,b]$ with 
\begin{align}
	\msL[\bs e_1 \tens \bs e_1 \tens \bs e_1] =  e^{2 X^1}, \quad \msL[\bs e_2 \tens \bs e_1 \tens \bs e_1] = 0. 
\end{align}
In particular, the stretch is not constant along the convected curve. The image of the coordinate curve $X^1 = c$ is the upper-half of the circle in the $x^1x^2$-plane centered at $(0,0)$ of radius $e^c$. The convected curve is parameterized by $X^2 \in [0,\pi]$, has constant stretch but nonzero curvature relative to the convected surface (i.e., the $x^1x^2$-plane), and it satisfies  
\begin{align}
	\msL[\bs e_2 \tens \bs e_2 \tens \bs e_2] = 0, \quad \msL[\bs e_1 \tens \bs e_2 \tens \bs e_2] = - e^{2 c}. 
\end{align}

\subsection{Strain energy and material symmetry} 
In our mathematical model, a Green-elastic body $\cl B$ with strain-gradient elastic surface $\cl S \subseteq \p \cl B$ is prescribed a strain energy of the form   
\begin{align}
	\Phi[\bs \chi] = \int_{\cl B} W(\bs C) \, dV + \int_{\cl S} U(\msE, \msK, \nabla \msE) \, dA, \label{eq:strainenergy}
\end{align}
where we omit listing the possible dependence of $W$ on $\bs X \in \cl B$ and of $U$ on $\bs Y \in \cl S$. We note that the strain energy is frame indifferent, i.e., it is invariant with respect to super-imposed rigid motions.\footnote{The fact that $U(\msE, \msK, \nabla \msE)$ is the most general form of a surface energy density depending on the first and second surface derivatives of the deformation that is frame indifferent was shown for planar $\cl S$ by Hilgers and Pipkin in Section 7 of \cite{HilgPip92a} { and for general, possible curved $\cl S$ by Steigmann in Section 2 of \cite{Steigmann18Lattice}.}} 

Although our main motivation for considering a strain-gradient elastic surface is when it forms part of the boundary of a substrate $\cl B$, we will treat material symmetry of $\cl S$ independently of that of $\cl B$. The notion of material symmetry for the energy density $W$ of the three-dimensional solid $\cl B$ is well-known, see \cite{Noll58, TruesdellNollNLFT}, so, we will limit our discussion to that of $\cl S$. 

Consider a material point with reference position $\bs Y_0 \in \cl S$. Our discussion of material symmetry for the surface energy per unit reference area $U$ at $\bs Y_0$ follows the framework introduced by Murdoch and Cohen \cite{MurdochCohen80, MurdCohenAdd81}. Their theory was later advocated for and reformulated by Steigmann and Ogden \cite{SteigOgden99} { and Steigmann \cite{Steigmann18Lattice}} using local coordinate parameterizations. We will follow their style of exposition, and unless specified otherwise, all quantities in what follows are evaluated at $\bs Y_0$. 

Let $\bs \lambda : \bbE^3 \rar \bbE^3$ be a rigid motion with deformation gradient $\bs R \in SO(3)$ satisfying 
 $\bs \lambda(\bs Y_0) = \bs Y_0.$ { For simplicity, we will limit our discussion to those $\bs R$ that fix the normal vector $\bs N$, $\bs R \bs N = \bs N$; see \cite{MurdochCohen80, MurdCohenAdd81, SteigOgden99, Steigmann18Lattice} for the more general theory.}
 We define a second reference surface $\cl S^* = \bigl \{
\bs Y^* = \bs \lambda^{-1}(\bs Y) \mid \bs Y \in \cl S
\bigr \}$.
It follows that $T_{\bs Y_0} \cl S^* = T_{\bs Y_0} \cl S$, at $\bs Y_0 \in \cl S^*$ the unit normal $\bs N^*$ to $\cl S^*$ satisfies $\bs N^* = \bs N = \bs R \bs N$, and 
\begin{align}
	\msR := \bs R |_{\bs T_{\bs Y_0} \cl S} : T_{\bs Y_0} \cl S \rar T_{\bs Y_0} \cl S
\end{align}
is a rotation.
A local parameterization on $\cl S^*$ is given by $\bs Y^* = \hat{\bs Y}^*(\theta^1, \theta^2) := \bs \lambda^{-1}(\hat{\bs Y}(\theta^1, \theta^2))$, so then 
\begin{align}
	\bs Y_{,\al} = \bs R \bs Y^*_{,\al}, \quad \bs Y^{,\al} = \bs R \bs Y^{*,\al}, \quad \al = 1, 2,
\end{align}
and at $\bs Y_0$, $\bs Y_{,\al} = \msR \bs Y^*_{,\al},$ $\bs Y^{,\al} = \msR \bs Y^{*,\al}$, for $\al = 1,2$. 

Let $\bs \chi: \bbE^3 \rar \bbE^3$ be a smooth invertible deformation of Euclidean space. Since a superimposed rigid motion does not affect the value of the surface energy, we will assume without loss of generality that 
\begin{align}
\bs \chi(\bs Y_0) = \bs Y_0 \mbox{ and at $\bs Y_0$, } \msF : T_{\bs Y_0} \cl S \rar T_{\bs Y_0} \cl S. \label{eq:defcond}
\end{align}
Following \cite{MurdCohenAdd81} we will impose the stronger requirement that 
\begin{align}
	\bs \chi(\bs Y_0) = \bs Y_0 \mbox{ and at $\bs Y_0$, } \bs F = \tensor{\mathsf{F}}{^\al_\beta} \bs Y_{,\al} \tens \bs Y^{,\beta} + \bs N \tens \bs N. \label{eq:defcondstrong}
\end{align}
In particular, it follows that at $\bs Y_0$, $\msF = \tensor{\mathsf{F}}{^\al_\beta} \bs Y_{,\al} \tens \bs Y^{,\beta}$ and $\bs n = \bs N = \bs F \bs N$. The deformation $\bs \chi^*(\cdot) := \bs \chi(\bs \lambda(\cdot))$ when restricted to $\cl S^*$ has the same image as $\bs \chi |_{\cl S}$, and in particular, the convected tangent vector fields are the same,
\begin{align}
	\bs y^*_{,\al} := \frac{\p}{\p \theta^\al} \bs \chi^*(\bs Y^*) = \frac{\p}{\p \theta^\al} \bs \chi(\bs Y) = \bs y_{,\al}.
\end{align}  
Then the surface deformation gradients of $\bs \chi$ and $\bs \chi^*$ at $\bs Y_0$ are related by 
\begin{align}
	\msF = \bs y_{\al} \tens \bs Y^{,\al} = \bs y^*_{\al} \tens \msR \bs Y^{*,\al} = \msF^* \msR^T \implies 
	\msC = \msR \msC^* \msR^T.  
\end{align}
The components of the first and second fundamental forms associated to $\cl S$ and $\cl S^*$ satisfy 
\begin{gather}
	\sG_{\al \beta} = \bs Y_{,\al} \cdot \bs Y_{,\beta} = \bs R \bs Y^*_{,\al} \cdot \bs R \bs Y^*_{\beta} = 
	\bs Y^*_{,\al} \cdot \bs Y^*_{\beta} = \sG^*_{\al \beta} \implies
	\msG = \msR \msG^* \msR^T, \\
	\sB_{\al \beta} = \bs Y_{,\al \beta} \cdot \bs N = \bs R \bs Y^*_{,\al \beta} \cdot \bs R \bs N^* = 
 \bs Y^*_{,\al \beta} \cdot \bs N^* = \sB^*_{\al \beta} \implies
 \msB = \msR \msB^* \msR^T,
\end{gather}
and thus, $\msE = \msR \msE^* \msR^T$. 
Since $\msb = \sanb_{\al \beta} \bs y^{,\al} \tens \bs y^{\beta}$ is the same for both deformations, we conclude that the relative normal curvature tensors $\msK$ and $\msK^*$ associated to $\bs \chi$ and $\bs \chi^*$ satisfy at $\bs Y_0$, 
\begin{gather}
	\msK = \msR \msK^* \msR^T.
\end{gather}
Finally, since the components of the first fundamental forms associated to $\bs \chi(\cl S)$ and $\bs \chi^*(\cl S^*)$ satisfy 
\begin{gather}
	\sg_{\al \beta} = \bs y_{\al} \cdot \bs y_{\beta} = \bs y^*_{\al} \cdot \bs y^*_{\beta} =: \sg^*_{\al \beta},
\end{gather}
we conclude that $\nabla \msE$ and $\nabla^* \msE^*$ at $\bs Y_0$ satisfy 
\begin{align}
\nabla \msE &= \frac{1}{2} \nabla_{\delta} \sg_{\al \beta} \bs Y^{,\al} \tens \bs Y^{,\beta} \tens \bs Y^{,\de} \\
&= \frac{1}{2} \nabla^*_{\delta} \sg_{\al \beta}^* \msR \bs Y^{*,\al} \tens \msR \bs Y^{*,\beta} \tens \msR \bs Y^{*,\de} \\
	&= \msR [ (\nabla^* \msE^*)^T \msR^T]^T \msR^T. 
\end{align}

We denote the surface energy per unit reference area relative to $\cl S^*$ by $U^*$. For the surface energy per unit mass to be independent of the reference surface used, we must have
\begin{gather}
	U^*(\msE^*, \msK^*, \nabla^* \msE^*) = U(\msE, \msH, \nabla\msE) \\ 
	= U \Bigl (\msR \msE^* \msR^T, \msR \msK^* \msR^T, \msR \bigl [ (\nabla^* \msE^*)^T \msR^T \bigr ]^T \msR^T \Bigr ). \label{eq:starrelation}
\end{gather}  

We now view $\bs \chi$ also as a deformation of $\cl S^*$, but we denote its values by $\bar{\bs \chi}$, 
\begin{align}
	\bar{\bs \chi}(\bs X) = {\bs \chi}(\bs X), \quad \bs X \in \bbE^3,
\end{align}
and the associated kinematic variables relative to $\cl S^*$ are denoted with an over-bar. We will now derive relationships between $\bar \msE$, $\bar \msK$, $\nabla^* \bar \msE$ and $\msE$, $\msK$, $\nabla \msE$. 

We first note that since ${\bs C} = \bar{\bs C}$, we immediately conclude that 
\begin{align}
	\msE &= \frac{1}{2} [ (\bs C - \bs I) \bs Y_{,\al} \cdot \bs Y_{,\beta} ] 
	\bs Y^{,\al} \tens \bs Y^{,\beta} \\
	&= \frac{1}{2} [ (\bar{\bs C} - \bs I) \bs Y_{,\al}^* \cdot \bs Y_{,\beta}^* ] 
	\bs Y^{*,\al} \tens \bs Y^{*,\beta} = \bar \msE. 
	 \label{eq:barC}
\end{align}
Let 
\begin{align}
	\mbox{Grad}\, {\bs C} = \frac{\p C_{ij}}{\p X^a} \bs e^i \tens \bs e^j \tens \bs e^a,
\end{align}
a third order tensor on $\bbR^3$. The identity $\sg_{\al \beta} = \bs C \bs Y_{,\al} \cdot \bs Y_{,\beta}$, the Gauss equations $\bs Y_{,\al \delta} = \tensor{\Gamma}{^\mu}{_{\al \delta}} \bs Y_{,\mu} + \sB_{\al \delta} \bs N$ on $\cl S$, and the symmetry of $\bs C$ imply that  
\begin{align}
\sg_{\al \beta, \delta} &= \mbox{Grad}\, \bs C[\bs Y_{,\al} \tens \bs Y_{,\beta} \tens \bs Y_{, \delta} ]
+ \bs C \bs Y_{,\al\delta} \cdot \bs Y_{,\beta} + \bs C \bs Y_{,\al} \cdot \bs Y_{,\beta \delta} \\
&= \mbox{Grad}\, \bs C[\bs Y_{,\al} \tens \bs Y_{,\beta} \tens \bs Y_{, \delta} ] + 
\tensor{\Gamma}{^\mu}{_{\al \delta}} \sg_{\mu \beta} + \sB_{\al \delta} \bs C \bs N \cdot \bs Y_{,\beta} \\
&\quad + \tensor{\Gamma}{^\mu}{_{\beta \delta}} \sg_{\mu \al} + \sB_{\beta \delta} \bs C \bs N \cdot \bs Y_{,\al}.
\end{align}
By \eqref{eq:defcondstrong}, we have $\bs C \bs N = \bs N$, and thus, at $\bs Y_0$, 
\begin{align}
	\nabla \msE &= \frac{1}{2} \mbox{Grad}\, \bs C 
	[\bs Y_{,\al} \tens \bs Y_{,\beta} \tens \bs Y_{, \delta} ]\bs Y^{,\al} \tens \bs Y^{,\beta} \tens \bs Y^{,\delta}. %&\quad + [\bs B \tens (\Pi \bs C \bs N)]^T + (\Pi \bs C \bs N) \tens \bs B.  
\end{align}
In particular, since $\bar{\bs C} = \bs C$ we conclude that 
\begin{align}
	\nabla^* \bar \msE = \nabla \msE. \label{eq:barnablaC}  
	 %+ [(\bs B^* - \bs B) \tens (\Pi \bs C \bs N )]^T + (\Pi \bs C \bs N) \tens (\bs B^* - \bs B).
\end{align}
Finally, as shown in Section 6 of \cite{SteigOgden99}, we have 
\begin{align}
	\bar \msK = \msK. \label{eq:barH}
\end{align}

Inspired by Murdoch and Cohen's extension of Noll's theory of material symmetry, we say that $\bs Y_0 \in \cl S$ is \emph{{ properly} symmetry related} to $\bs Y_0 \in \cl S^*$ if the mechanical responses to the arbitrary deformation $\bs \chi$ are identical, i.e.,\footnote{ Our use of ``proper" is not to indicate that any other theory is improper. It is simply to indicate that we are considering rigid motions $\bs R$ that fix the normal vector and induce a rotation $\msR$ of the tangent plane $T_{\bs Y_0} \cl S^* = T_{\bs Y_0} \cl S$.} 
\begin{align}
		U(\msE, \msK, \nabla \msE) = U^*(\bar \msE, \bar \msK, \nabla^* \bar \msE).
\end{align}
By \eqref{eq:starrelation}, \eqref{eq:barC}, \eqref{eq:barnablaC}, \eqref{eq:barH}, this requirement leads to the following definition: the rotation $\msR$ is in the \emph{{proper} symmetry set} of $\bs Y_0$ relative to $\cl S$ if for every smooth, invertible deformation $\bs \chi: \bbE^3 \rar \bbE^3$ satisfying \eqref{eq:defcondstrong}, we have 
\begin{gather}
	U(\msE, \msK, \nabla \msE) 
	= U \Bigl (\msR \msE \msR^T, \msR \msK \msR^T,
	\msR \bigl [ \nabla\msE^T \msR^T \bigr]^T \msR^T \Bigr ). \label{eq:symmetryset}
\end{gather}
As in the case of the standard theory of Noll \cite{Noll58}, one can verify that the {proper} symmetry set of $\bs Y_0$ relative to $\cl S$ is a subgroup of the group of rotations of $T_{\bs Y_0} \cl S$. In the case that the {proper} symmetry set of $\bs Y_0$ relative to $\cl S$ equals the group of rotations of $T_{\bs Y_0} \cl S$, we say that the surface energy density $U$ is \emph{{properly} hemitropic} at $\bs Y_0$. 

\subsection{Field equations}

Following \cite{SteigOgden99}, the field equations for the body $\cl B$ with strain-gradient elastic surface $\cl S \subseteq \p \cl B$ are defined to be the Euler-Lagrange equations for the Lagrangian energy functional 
\begin{align}
	\cl A[\bs \chi] = \Phi[\bs \chi] + V[\bs \chi] \label{eq:action}
\end{align}
where $V[\bs \chi]$ is the potential energy associated to the applied forces.\footnote{See \cite{Steigmann18Lattice} for the case of a stand-alone elastic surface with no substrate.} In this work, we assume that the G\^ateaux derivative of the load potential takes the form 
\begin{align}
	\dot V = - \int_{\cl B} \bs f \cdot \bs u \, dV - \int_{\cl S} \bs t \cdot \bs u \, dA,
\end{align}
where $\bs f$ is a prescribed external body force on $\cl B$ and $\bs t$ is a prescribed boundary traction on $\cl S$. 

Let $\bs \chi(\cdot ; \eps)$ be a one-parameter family of deformations of $\cl B$ such that 
 $\bs \chi(\cdot; \eps) |_{\p \cl B \backslash \cl S} = \bs \chi_0(\cdot)$, and denote  
\begin{align}
	\dot{} := \frac{d}{d\eps} \Big |_{\eps = 0}, \quad \bs u := \dot{\bs \chi}(\cdot; \eps). 
\end{align}
Then $\bs u |_{\p \cl B}$ vanishes to first order on $\p \cl B \backslash \cl S$. 
Using the chain rule and integration by parts we have the classical identity   
\begin{align}
	\int_{\cl B} \dot W \, dV = \int_{\p \cl B} \bs P \bs N \cdot \bs u \, dA - \int_{\cl B} \Div \bs P \cdot \bs u\, dV \label{eq:Wvariation}
\end{align}
where $\bs P = \tensor{P}{_i^a} \bs e^i \tens \bs e_a$ is the Piola stress with $\tensor{P}{_i^a} = \frac{\p W}{\p \tensor{F}{^i_a}}$ and $\Div \bs P = \bigl ( \p_{X^a} \tensor{P}{_i^a} \big ) \bs e^i$.
Using the chain rule we have that  
\begin{gather}
	\int_{\cl S} \dot U \, dA = \int_{\cl S} \Bigl ( \msT^\al \cdot \bs u_{,\al} + \msM^{\al \beta} \cdot \bs u_{,\al \beta} \Bigr ) dA, \label{eq:variation}\\
	\msT^\al := \frac{\p U}{\p y^k_{,\al}} \bs e^k, \quad 
	\msM^{\al \beta} := \frac{\p U}{\p y^{k}_{,\al \beta}} \bs e^k.
\end{gather}
We define surface stress vectors $\msP^\al$ by
\begin{align}
	\quad \msP^\al := \msT^\al - G^{-1/2} (G^{1/2} \msM^{\al \beta})_{,\beta}, 
\end{align}
and use \eqref{eq:Wvariation}, \eqref{eq:variation} and integration by parts to obtain the Euler-Lagrange equations associated to the Lagrangian energy functional \eqref{eq:action}, 
\begin{align}
\begin{split}
	&\Div \bs P + \bs f = \bs 0, \quad \mbox{on } \cl B, \\
	&\bs P \bs N = G^{-1/2} (G^{1/2} \msP^\al )_{,\al} + \bs t, \quad \mbox{on } \cl S, \\
	&\bs \chi(\bs X) = \bs \chi_0(\bs X), \quad \mbox{on } \p \cl B \backslash \cl S. 
\end{split}\label{eq:fieldequations}
\end{align}

\section{Small strain models}

Our principle motivation for modeling an elastic solid with strain-gradient elastic boundary surface is the study of brittle fracture. In this setting (to be discussed more in the following section), the surface $\cl S$ possessing strain-gradient surface elasticity will be the crack front and strains will be linearized, motivating the introduction of a small-strain surface energy density. In this section, we present a model uniform, properly hemitropic, small-strain surface energy density that requires the same material constants (with the same physical interpretations) as found in the narrower Steigmann-Ogden theory. In contradistinction, the surface energy incorporates the surface's resistance to geodesic distortion and satisfies the strong ellipticity condition. Moreover, the surface energy density may be viewed as a geometric generalization of that introduced and advocated for by Hilgers and Pipkin in \cite{HilgPip92a, HilgPip96}.   

\subsection{Hilgers-Pipkin surface energy}

In what follows indices are raised and lowered using the reference metric $\msG$, but we note that $\bs y^{,\al}$ are the dual vector fields to $\bs y_{,\al}$ and are not given by $\bs y_{,\beta} (\sG^{-1})^{\al \beta}$. For the surface $\cl S$, we propose the uniform, properly hemitropic surface energy density 
\begin{align}
	U &= \frac{\lambda_s}{2} (\tensor{\sE}{^\al_\al} )^2 + \mu_s \sE_{\al \beta} \sE^{\al \beta}
	+ \frac{\zeta}{2} \Bigl [ (\tensor{\sK}{^\al_\al} )^2 + (\sg^{-1})^{\mu \nu} \tensor{\sL}{_\mu_\al^\al} \tensor{\sL}{_\nu_\beta^\beta} \Bigr ] \\
	&\quad + \eta \Bigl [ \tensor{\sK}{_\al_\beta}\sK^{\al \beta} + (\sg^{-1})^{\mu \nu} \tensor{\sL}{_\mu_\al_\beta} \tensor{\sL}{_\nu^\alpha^\beta} \Bigr ]. \label{eq:surfaceen}
\end{align} 
Here $\lambda_s$, $\mu_s$, $\zeta$ and $\eta$ are positive numbers that can be interpreted as the surface Lam\'e constants and pure bending moduli.

In the case that $\cl S$ is contained in a plane with flat coordinates $(\theta^1, \theta^2)$, we have 
\begin{gather}
	\sE_{\al \beta} = \frac{1}{2}(\sg_{\al \beta} - \delta_{\al \beta}), \quad \sL_{\mu \al \beta} = \bs y_{,\mu} \cdot \bs y_{,\al \beta}, \\
	\bs y_{,\al \beta} = \sL_{\mu \al \beta} \bs y^{,\mu} + \sK_{\al \beta} \bs n, \quad
	\sum_{\al = 1}^2 \bs y_{,\al \al} = \tensor{\sL}{_\mu_\al^\al} \bs y^{,\mu} + \tensor{\sK}{^\al_\al} \bs n, \\
	\Bigl |\sum_{\al = 1}^2 \bs y_{,\al\al} \Bigr |^2 = (\tensor{\sK}{^\al_\al} )^2 + (\sg^{-1})^{\mu \nu} \tensor{\sL}{_\mu_\al^\al} \tensor{\sL}{_\nu_\beta^\beta}, \\
	\sum_{\al, \beta = 1}^2 |\bs y_{,\al \beta}|^2 = \tensor{\sK}{_\al_\beta}\sK^{\al \beta} + (\sg^{-1})^{\mu \nu} \tensor{\sL}{_\mu_\al_\beta} \tensor{\sL}{_\nu^\alpha^\beta},
\end{gather}    
and \eqref{eq:surfaceen} becomes 
\begin{align}
	U = \frac{\lambda_s}{2} (\tensor{\sE}{^\al_\al} )^2 + \mu_s \sE_{\al \beta} \sE^{\al \beta}
	+ \frac{\zeta}{2} \Bigl |\sum_{\al = 1}^2 \bs y_{,\al\al} \Bigr |^2
	+ \eta \sum_{\al, \beta = 1}^2 |\bs y_{,\al \beta}|^2. \label{eq:HPen} 
\end{align}

Up to a choice of constants, the surface energy density \eqref{eq:HPen} is precisely that introduced by Hilgers and Pipkin in \cite{HilgPip92a}, and therefore, we refer to \eqref{eq:surfaceen} as a \emph{small-strain Hilgers-Pipkin surface energy}.  In \cite{HilgPip92a, HilgPip96}, Hilgers and Pipkin advocated for the use of \eqref{eq:HPen} over the classical surface energy
\begin{align}
	U &= \frac{\lambda_s}{2} (\tensor{\sE}{^\al_\al} )^2 + \mu_s \sE_{\al \beta} \sE^{\al \beta}
	+ \frac{\zeta}{2} \tensor{\sK}{^\al_\al} + \eta \sK_{\al \beta} \sK^{\al \beta} \\
	&= \frac{\lambda_s}{2} (\tensor{\sE}{^\al_\al} )^2 + \mu_s \sE_{\al \beta} \sE^{\al \beta} + \frac{\zeta}{2} \Bigl |\sum_{\al = 1}^2 (\bs n \cdot \bs y_{,\al\al}) \Bigr |^2
	+ \eta \sum_{\al, \beta = 1}^2 (\bs n \cdot \bs y_{,\al \beta})^2 \label{eq:classen}
\end{align}
on the basis of \eqref{eq:HPen} being analytically simpler than \eqref{eq:classen}.  Indeed, with little effort one sees that for \eqref{eq:HPen}, 
\begin{align}
	\msT^\al &= (\lambda_s \tensor{\sE}{^\gamma_\gamma} \delta^{\alpha \beta} + 2 \mu_s \sE^{\al \beta}) \bs y_{,\beta}, \label{eq:HPent}\\
	\msM^{\al \beta} &= \zeta \Bigl (\sum_{\gamma} \bs y_{,\gamma \gamma} \Bigr ) \delta^{\al \beta} + 2 \eta \bs y_{,\al \beta}. \label{eq:HPenm}
\end{align}
Moreover, it is simple to see that \eqref{eq:HPen} satisfies the \textit{strong ellipticity condition}, 
\begin{gather}
\forall 
(\sa_1, \sa_2) \in \bbR^2 \backslash\{\bs (0,0)\}, \bs b \in \bbR^3 \backslash \{\bs 0\}, \quad	\sa_{\al} \sa_{\beta} \bs b \cdot \Bigl ( \bs C^{\al \be \delta \gamma} \sa_{\de} \sa_{\gamma} \bs b \Bigr ) > 0,  \label{eq:strongellipt} \\
	\bs C^{\al \beta \gamma \delta} := \frac{\p^2 U}{\p y^{i}_{,\al \beta} \p y^{j}_{,\de \gamma}} \bs e^i \tens \bs e^j, 
\end{gather}
while \eqref{eq:classen} does not (see \eqref{eq:strong} and \eqref{eq:LHcondition} below). Physically, this may be viewed as a consequence of the surface energy \eqref{eq:surfaceen} incorporating the surface's resistance to geodesic distortion via also including dependence on the tensor $\msL$.  

In general, for \eqref{eq:surfaceen} we have  
\begin{align}
	\msT^\al = \frac{\p U}{\p \bs y_{,\al}} = 
	\frac{\p U}{\p \sE_{\beta \gamma}} \frac{\p \sE_{\beta \gamma}}{\p \bs y_{,\al}} + \frac{\p U}{\p \sK_{\delta \nu}} \frac{\p \sK_{\delta \nu}}{\p \bs y_{,\al}} + 
	\frac{\p U}{\p \sL_{\beta \gamma \delta}} \frac{\p \sL_{\beta \gamma \delta}}{\p \bs y_{,\al}}.  
\end{align}
Using 
\begin{gather}
	\frac{\p \sE_{\beta \gamma}}{\p \bs y_{,\al}} = \frac{1}{2}\Bigl (\tensor{\delta}{^\al_\beta} \bs y_{,\gamma} + \tensor{\delta}{^\al_\gamma}  \bs y_{,\beta} \Bigr ), \\
	\frac{\p \sK_{\mu \nu}}{\p \bs y_{,\al}} = - \tensor{\gamma}{^\al_\mu_\nu} \bs n, \quad 
		\frac{\p \sK_{\mu \nu}}{\p \bs y_{,\al \beta}} = \frac{1}{2} (\tensor{\delta}{^\al_\mu}\tensor{\delta}{^\beta_\nu} + \tensor{\delta}{^\al_\nu}\tensor{\delta}{^\beta_\mu} ) \bs n, \\
	\frac{\p \sL_{\beta \gamma \mu}}{\p \bs y_{,\al}} = \tensor{\delta}{^\al_\beta} \bs y_{\gamma \mu} - 
	(\tensor{\delta}{^\al_\nu} \bs y_{,\beta} + \tensor{\delta}{^\al_\beta}  \bs y_{,\nu}) \tensor{\Gamma}{^\nu_\gamma_\mu}, \\
	\frac{\p \sL_{\gamma \mu \sigma}}{\p \bs y_{,\al \beta}} = \frac{1}{2} (
	\tensor{\delta}{^\al_\mu} \tensor{\delta}{^\beta_\sigma} + \tensor{\delta}{^\al_\sigma} \tensor{\delta}{^\beta_\mu}) \bs y_{\gamma},
\end{gather}
we readily compute that
\begin{comment} 
\begin{align}
	\frac{\p U}{\p \sE_{\beta \gamma}} \frac{\p \sE_{\beta \gamma}}{\p \bs y_{,\al}} &= 
	\Bigl [
	\lambda_s \tensor{\sE}{^\mu_\mu} (\sG^{-1})^{\al \gamma} + 2 \mu_s \sE^{\al \gamma} \\
	&\qquad - (\sg^{-1})^{\mu \al} (\sg^{-1})^{\nu \gamma} (
	\zeta \tensor{\sL}{_\mu_\delta^\delta} \tensor{\sL}{_\nu_\delta^\delta}
	+ 2\eta \tensor{\sL}{_\mu_\delta_\sigma} \tensor{\sL}{_\nu^\delta^\sigma}
	)
	\Bigr ] \bs y_{,\gamma}, \\
	\frac{\p U}{\p \sK_{\delta \nu}} \frac{\p \sK_{\delta \nu}}{\p \bs y_{,\al}} &= 
	-\Bigl [
	\zeta \tensor{\sK}{^\mu_\mu} (\sG^{-1})^{\delta \nu} \tensor{\gamma}{^\al_\delta_\nu} + 2 \eta \sK^{ \delta \nu} \tensor{\gamma}{^\al_\delta_\nu} \Bigr ]\bs n,  \\
	\frac{\p U}{\p \sL_{\beta \gamma \delta}} \frac{\p \sL_{\beta \gamma \delta}}{\p \bs y_{,\al}} &=
	\Bigl [
	\zeta (\sg^{-1})^{\mu \al} \tensor{\sL}{_\mu_\nu^\nu} (G^{-1})^{\gamma \delta} + 2 \eta (g^{-1})^{\mu \al} \tensor{\sL}{_\mu^\gamma^\delta} \Bigr ] \bs y_{,\gamma \delta} \\
	&\quad - \zeta \Bigl [ (\sg^{-1})^{\mu \al} \tensor{\sL}{_\mu_\nu^\nu} \tensor{\Gamma}{^\beta_\nu^\nu}
	+ (\sg^{-1})^{\mu \beta} \tensor{\sL}{_\mu_\nu^\nu} \tensor{\Gamma}{^\al_\nu^\nu}
	\Bigr ] \bs y_{,\beta} \\
	&\quad - 2 \eta \Bigl [
	(\sg^{-1})^{\mu \al} \tensor{\sL}{_\mu^\delta^\sigma} \tensor{\Gamma}{^\beta_\delta_\sigma}
	+ (\sg^{-1})^{\mu \beta} \tensor{\sL}{_\mu^\delta^\sigma} \tensor{\Gamma}{^\al_\delta_\sigma}
	\Bigr ] \bs y_{,\beta}, 
\end{align}
and thus, \end{comment}
\begin{align}
		\msT^\al &= 
		\Bigl [
		\lambda_s \tensor{\sE}{^\mu_\mu} (\sG^{-1})^{\al \gamma} + 2 \mu_s \sE^{\al \gamma} \\
		&\qquad - (\sg^{-1})^{\mu \al} (\sg^{-1})^{\nu \gamma} (
		\zeta \tensor{\sL}{_\mu_\delta^\delta} \tensor{\sL}{_\nu_\sigma^\sigma}
		+ 2\eta \tensor{\sL}{_\mu_\delta_\sigma} \tensor{\sL}{_\nu^\delta^\sigma}
		)
		\Bigr ] \bs y_{,\gamma}, \\
		&\quad  
		-\Bigl [
		\zeta \tensor{\sK}{^\mu_\mu} (\sG^{-1})^{\delta \nu} \tensor{\gamma}{^\al_\delta_\nu} + 2 \eta \sK^{ \delta \nu} \tensor{\gamma}{^\al_\delta_\nu} \Bigr ]\bs n,  \\ 
		&\quad +
		\Bigl [
		\zeta (\sg^{-1})^{\mu \al} \tensor{\sL}{_\mu_\nu^\nu} (\sG^{-1})^{\gamma \delta} + 2 \eta (g^{-1})^{\mu \al} \tensor{\sL}{_\mu^\gamma^\delta} \Bigr ] \bs y_{,\gamma \delta} \\
		&\quad - \zeta \Bigl [ (\sg^{-1})^{\mu \al} \tensor{\sL}{_\mu_\nu^\nu} \tensor{\Gamma}{^\beta_\delta^\delta}
		+ (\sg^{-1})^{\mu \beta} \tensor{\sL}{_\mu_\nu^\nu} \tensor{\Gamma}{^\al_\delta^\delta}
		\Bigr ] \bs y_{,\beta} \\
		&\quad - 2 \eta \Bigl [
		(\sg^{-1})^{\mu \al} \tensor{\sL}{_\mu^\delta^\sigma} \tensor{\Gamma}{^\beta_\delta_\sigma}
		+ (\sg^{-1})^{\mu \beta} \tensor{\sL}{_\mu^\delta^\sigma} \tensor{\Gamma}{^\al_\delta_\sigma}
		\Bigr ] \bs y_{,\beta}. 
\end{align}
and 
\begin{align}
	\msM^{\al \beta} &= \Bigl [
	\zeta \tensor{\sK}{^\mu_\mu} (\sG^{-1})^{\al \beta} + 2 \eta \sK^{\al \beta}  \Bigr ] \bs n \\
	&\quad + \Bigl [
	\zeta (\sg^{-1})^{\mu \gamma} \tensor{\sL}{_\mu_\nu^\nu} (\sG^{-1})^{\al \beta}
	+ 2 \eta (\sg^{-1})^{\mu \gamma} \tensor{\sL}{_\mu^\al^\beta} 
	\Bigr ] \bs y_{,\gamma}.
\end{align}

To see that the strong ellipticity condition is satisfied for \eqref{eq:surfaceen}, we compute 
\begin{align}
	\bs C^{\al \beta \mu \nu} = \Bigl (
	\zeta (\sG^{-1})^{\al \beta} (\sG^{-1})^{\mu \nu} + \eta [(\sG^{-1})^{\al \mu}(\sG^{-1})^{\beta \nu} + (\sG^{-1})^{\al \nu}(\sG^{-1})^{\beta \mu}] \Bigr ) \bs I, 
\end{align}
and thus, for all 
$(\sa_1, \sa_2) \in \bbR^2 \backslash\{\bs (0,0)\}$ and $\bs b \in \bbR^3 \backslash \{\bs 0\}$,
\begin{align}
	\sa_{\al} \sa_{\beta} \bs b \cdot \Bigl ( \bs C^{\al \be \delta \gamma} \sa_{\de} \sa_{\gamma} \bs b \Bigr ) = 
	(\zeta + 2 \eta) [(\sG^{-1})^{\al \beta} \sa_\al \sa_\beta]^2 |\bs b|^2 > 0. \label{eq:strong}
\end{align}
For the classical surface energy \eqref{eq:classen} of Steigmann-Ogden type, we have 
\begin{align}
		\bs C^{\al \beta \mu \nu} = 
	\Bigl ( \zeta (\sG^{-1})^{\al \beta} (\sG^{-1})^{\mu \nu} + \eta [(\sG^{-1})^{\al \mu}(\sG^{-1})^{\beta \nu} + (\sG^{-1})^{\al \nu}(\sG^{-1})^{\beta \mu}] \Bigr ) \bs n \tens \bs n,
\end{align}
and thus, for all 
$(\sa_1, \sa_2) \in \bbR^2$ and $\bs b \in \bbR^3$,
\begin{align}
	\sa_{\al} \sa_{\beta} \bs b \cdot \Bigl ( \bs C^{\al \be \delta \gamma} \sa_{\de} \sa_{\gamma} \bs b \Bigr ) = 
	(\zeta + 2 \eta) [(\sG^{-1})^{\al \beta} \sa_\al \sa_\beta]^2 (\bs n \cdot \bs b)^2. \label{eq:LHcondition}
\end{align}
In particular, \eqref{eq:LHcondition} shows that the surface energy \eqref{eq:classen} satisfies the associated Legendre-Hadamard condition { (see \cite{Graves39, HilgersPipkin92Graves})
\begin{gather}
	\forall 
	(\sa_1, \sa_2) \in \bbR^2 \backslash\{\bs (0,0)\}, \bs b \in \bbR^3 \backslash \{\bs 0\}, \quad	\sa_{\al} \sa_{\beta} \bs b \cdot \Bigl ( \bs C^{\al \be \delta \gamma} \sa_{\de} \sa_{\gamma} \bs b \Bigr ) \geq 0,  \label{eq:LegHad}
\end{gather}
}but not the strong ellipticity condition since the right side of \eqref{eq:LHcondition} is 0 for $\bs b \neq \bs 0$ and orthogonal to $\bs n$.

\subsection{Parameter values}
Viewing $\cl S$ as the midsurface of a flat, thin, uniform, isoptropic strip with thickness $h$ and Lam\'e parameters $\la_l$, $\mu_l$, the works \cite{HilgPip96, Steig13} suggest the values 
\begin{gather}
	\lambda_s = \frac{2 \la_l \mu_l}{\la_l + 2 \mu_l} h, \quad \mu_s = \mu_l h, \quad
	\zeta = \frac{2 \la_l \mu_l}{\la_l + 2 \mu_l} \frac{h^3}{24}, \quad \eta = \frac{h^3}{24} \mu_l. \label{eq:values}  
\end{gather}      
For \eqref{eq:values}, the Hilgers-Pipkin surface energy \eqref{eq:surfaceen} agrees with Koiter's classical shell energy \cite{Koit66} for homogeneous plane strain and pure bending deformations of a plate. 

We note that Koiter's shell energy has been derived as the leading order model in small thickness from classical nonlinear elasticity (see \cite{HilgPip96, Steig13}), while, to the author's knowledge, \eqref{eq:surfaceen} has not. { Moreover, most work deriving surface energies as a small thickness approximation to a three-dimensional strain energy concerns materials having reflection symmetry with respect to a midsurface, and the resulting models do not contain strain-gradients. A notable exception is the work \cite{Steigmann2012KoiterExtension} that derives surface energies from parent strain energies exhibiting aribrary symmetry, and materials without reflection symmetry yield models that do include strain gradients. Deriving the specific Hilgers-Pipkin surface energy \eqref{eq:surfaceen} or related strain-gradient surface energies from three-dimensional nonlinear elasticity or strain-gradient elasticity will be addressed in future work.}  

\subsection{Linearized equations} 
We now compute the linearization of \eqref{eq:field} about the reference configuration. For the substrate, we adopt a classical quadratic, isotropic energy density  
\begin{align}
	W = \frac{\lambda}{2} (\tensor{E}{^i_i})^2 + \mu E_{ij} E^{ij}.
\end{align}
Here indices are raised using the flat metric on $\bbR^3$, and $\lambda$ and $\mu$ are the Lam\'e constants for the three-dimensional solid. For the surface $\cl S$, the surface energy density is given by \eqref{eq:surfaceen}.

Let $\bs u : \cl B \rar \R^3$ be a displacement field such that $\bs u |_{\p \cl B \backslash \cl S} = \bs 0$ and 
\begin{gather}
\sup_{\bs X \in \cl B} \Bigl [ |\bs u(X)| + |\mbox{Grad}\, \bs u(\bs X)| \Bigr ] + 
\sum_{\al, \beta = 1}^2 \sup_{\bs Y \in \cl S} |\bs u_{,\al \beta}(\bs Y)| \leq \delta_0. \label{eq:infintesimaldisplacement}
\end{gather}
Assume that the body force $\bs f$, boundary traction $\bs t$, and Dirichlet condition $\bs \chi_0$ satisfy
\begin{align}
	|\bs f| = O(\delta_0), \quad |\bs t| = O(\delta_0), \quad |\bs \chi_0 - \mbox{Id}| = O(\delta_0). 
\end{align}
If
$
\bs \chi(\bs X) = \bs X +\bs u(\bs X),
$
then $E_{ij} = \veps_{ij} + O(\delta_0^2)$, $\sE_{\al \beta} = \eps_{\al \beta} + O(\delta_0^2)$, and $\sK_{\al \beta} = \sk_{\al \beta} + O(\delta_0^2)$ where  
\begin{align}
	\veps_{ij} = \frac{1}{2} \Bigl (
	\bs e_i \cdot \frac{\p \bs u}{\p X^j}+ \bs e_j \cdot \frac{\p \bs u}{\p X^i}
	\Bigr ) = O(\delta_0),
\end{align}
and on $\cl S$, 
 \begin{align}
 	 \eps_{\al \beta} = \frac{1}{2} \bigl ( \bs Y_{,\al} \cdot \bs u_{,\beta} + \bs Y_{,\al} \cdot \bs u_{,\beta} \bigr ) = O(\delta_0), \quad 
 	\sk_{\al \beta} = \bs N \cdot \bs u_{;\al \beta} = O(\delta_0), \label{eq:EKincrement}
 \end{align}
see (3.12) in \cite{SteigOgden99}. Now we observe that
$	\sL_{\al \beta \delta} = \sanl_{\al \beta \delta} + O(\delta_0^2),$ with 
\begin{align}
	\sanl_{\al \beta \delta} &= \bs Y_{,\al} \cdot \bs u_{,\beta \delta} + \bs Y_{,\beta \delta} \cdot \bs u_{,\al}  - \tensor{\Gamma}{^\mu_\beta_\delta}\bigl ( \bs Y_{,\al} \cdot \bs u_{,\mu} + \bs Y_{,\mu} \cdot \bs u_{,\al} \bigr )\\
		&= \bs Y_{,\al} \cdot \bs u_{;\beta \delta} + \bs Y_{;\beta \delta} \cdot \bs u_{,\al} \\
	&= \bs Y_{,\al} \cdot \bs u_{;\beta \delta} + (\bs N \cdot \bs u_{,\al})\sB_{\beta \delta} = O(\delta_0). \label{eq:Lincrement}
\end{align}
\begin{comment}
Then $W = W_L + O(\delta^3)$ and $U = U_L + O(\delta^3)$ with 
\begin{align}
	W_L &= \frac{\lambda}{2} (\tensor{\veps}{^i_i})^2 + \mu \veps_{ij} \veps^{ij}, \\
	U_L &= \frac{\lambda_s}{2} (\tensor{\eps}{^\al_\al} )^2 + \mu_s \eps_{\al \beta} \eps^{\al \beta}
	+ \frac{\zeta}{2} \Bigl [ (\tensor{\sk}{^\al_\al} )^2 + \tensor{\sanl}{_\mu_\al^\al} \tensor{\sanl}{^\mu_\beta^\beta} \Bigr ] \\
	&\quad + \eta \Bigl [ \tensor{\sk}{_\al_\beta}\sk^{\al \beta} + \tensor{\sanl}{_\mu_\al_\beta} \tensor{\sanl}{^\mu^\alpha^\beta} \Bigr ] = O(\delta).
\end{align}
\end{comment}
Then $\msT^\al = \mst^\al + O(\delta_0^2)$ and $\msM^{\al \beta} = \msm^{\al \beta} + O(\delta_0^2)$ where 
\begin{comment}
\begin{gather}
	\frac{\p}{\p \bs u_{,\al}} \tensor{\eps}{^\mu_\mu} = (G^{-1})^{\al \mu} \bs Y_{,\mu}, \quad
	\frac{\p}{\p \bs u_{,\al}} \eps_{\mu \nu} = \frac{1}{2} \Bigl [
	\bs Y_{,\mu} \tensor{\delta}{^\al_\nu} + \bs Y_{,\nu} \tensor{\delta}{^\al_\mu}
	\Bigr ], \\
	\frac{\p}{\bs u_{,\al \beta}} \tensor{\sk}{^\mu_\mu} = (G^{-1})^{\al \beta} \bs N, \quad 
	 	\frac{\p}{\p \bs u_{,\al\beta}} \sk_{\mu \nu} = \frac{1}{2} \Bigl [
	 	\tensor{\delta}{^\al_\mu} \tensor{\delta}{^\beta_\nu} + \tensor{\delta}{^\al_\nu} \tensor{\delta}{^\beta_\mu}
	 	\Bigr ] \bs N, \\
	 \frac{\p}{\p \bs u_{,\al}} \tensor{\sanl}{_\mu_\nu_\xi} = 
	 - \tensor{\Gamma}{^\al_\nu_\xi} \bs Y_{,\mu} + \tensor{\delta}{^\al_\mu} \sB_{\nu \xi} \bs N, \\ 	
	 \frac{\p}{\p \bs u_{,\al\beta}} \tensor{\sanl}{_\mu_\nu^\nu} = (G^{-1})^{\al \beta} \bs Y_\mu, 
	 \quad \frac{\p}{\p \bs u_{,\al\beta}} \sanl_{\mu \nu \xi} = \frac{1}{2} \Bigl [
	 \tensor{\delta}{^\al_\nu} \tensor{\delta}{^\beta_\xi} + \tensor{\delta}{^\al_\xi} \tensor{\delta}{^\beta_\nu}
	 \Bigr ] \bs Y_{\mu}.
\end{gather} 
\end{comment}
\begin{align}
	\mst^\al &:= \Bigl [
	\lambda_s \tensor{\eps}{^\mu_\mu} (\sG^{-1})^{\al \gamma} + 2 \mu_s \eps^{\al \gamma} \Bigr ] \bs Y_{,\gamma}  
	-\Bigl [
	\zeta \tensor{\sk}{^\mu_\mu} \tensor{\Gamma}{^\al_\nu^\nu} + 2 \eta \sk^{ \delta \nu} \tensor{\Gamma}{^\al_\delta_\nu} \Bigr ]\bs N,  \\ 
	&\quad + \Bigl [
	\zeta \tensor{\sanl}{^\al_\nu^\nu} \tensor{\sB}{^\delta_\delta} + 2 \eta \sanl^{\al \delta \sigma} \sB_{\delta \sigma}  
	\Bigr ] \bs N 
	 - \Bigl [
	 \zeta \tensor{\sanl}{^\beta_\nu^\nu} \tensor{\Gamma}{^\al_\delta^\delta} + 2 \eta \sanl^{\beta \delta \sigma} \tensor{\Gamma}{^\al_\delta_\sigma}
	 \Bigr ] \bs Y_{,\beta}, \\
	 \msm^{\al \beta} &:= \Bigl [
	 \zeta \tensor{\sk}{^\mu_\mu} (\sG^{-1})^{\al \beta} + 2 \eta \sk^{\al \beta}  \Bigr ] \bs N  + \Bigl [
	 \zeta \tensor{\sanl}{^\gamma_\nu^\nu} (\sG^{-1})^{\al \beta}
	 + 2 \eta  \tensor{\sanl}{^\gamma^\al^\beta} 
	 \Bigr ] \bs Y_{,\gamma}. 
\end{align}

The linearization of \eqref{eq:fieldequations} about the reference configuration is obtained by omitting the $O(\delta_0^2)$ terms from $\bs P$, $\msT^\al$ and $\msM^{\al \beta}$, yielding 
\begin{align}
	\begin{split}
		&\Div \bs \sigma + \bs f = \bs 0, \quad \mbox{on } \cl B, \\
		&\bs \sigma \bs N = G^{-1/2} (G^{1/2} \msp^\al )_{,\al} + \bs t, \quad \mbox{on } \cl S, \\
		&\bs u = \bs 0, \quad \mbox{on } \p \cl B \backslash \cl S,
	\end{split}\label{eq:linfieldequations}
\end{align}
where $\bs \sigma = \lambda (\tr \bs \veps) \bs I + 2\mu \bs \veps$ and $\msp ^\al = \mst^\al - G^{-1/2} (G^{1/2} \msm^{\al \beta})_{,\beta}$. We observe that solutions to the linearized equations \eqref{eq:linfieldequations} are critical points of the energy functional
\begin{align}
	\cl A_L[\bs u] &= \int_{\cl B} \Bigl [ \frac{\lambda}{2} (\tensor{\veps}{^i_i})^2 + \mu \veps_{ij} \veps^{ij} \Bigr ] dV - \int_{\cl B} \bs f \cdot \bs u \, dV
	 + \int_{\cl S} \Bigl [
	\frac{\lambda_s}{2} (\tensor{\eps}{^\al_\al} )^2 + \mu_s \eps_{\al \beta} \eps^{\al \beta} \Bigr ] dS \\
	&\quad 
	+ \int_{\cl S} \Bigl [\frac{\zeta}{2} \Bigl ( (\tensor{\sk}{^\al_\al} )^2 + \tensor{\sanl}{_\mu_\al^\al} \tensor{\sanl}{^\mu_\beta^\beta} \Bigr )  + \eta \Bigl ( \tensor{\sk}{_\al_\beta}\sk^{\al \beta} + \tensor{\sanl}{_\mu_\al_\beta} \tensor{\sanl}{^\mu^\alpha^\beta} \Bigr )
	\Bigr ] dA - \int_{\cl S} \bs t \cdot \bs u \, dA 
\end{align}
over the set of $\bs u$ satisfying $\bs u |_{\p \cl B \backslash \cl S} = \bs 0$. 

In what follows, $\tensor{\bs g}{_{,\alpha}^\beta} = \delta^{\beta \mu} \bs g_{,\alpha \mu}$ and $\bs g^{,\alpha \beta} = \delta^{\alpha \lambda} \delta^{\beta \mu} \bs g_{,\lambda \mu}$. For the case of the classical Hilgers-Pipkin surface energy \eqref{eq:HPen}, we see from \eqref{eq:HPent} and \eqref{eq:HPenm} that 
\begin{align}
	\mst^\al = (\lambda_s \tensor{\eps}{^\gamma_\gamma} \delta^{\al \beta} + 2 \mu_s \eps^{\al \beta}) \bs Y_{,\beta}, \quad 
	\msm^{\al \beta} = \zeta \delta^{\al \beta} \tensor{\bs u}{_{,\gamma}^\gamma} + 2 \eta \bs u^{,\al \beta} 
\end{align}
Writing $\bs u = \msu + \su^3 \bs N = \su^\gamma \bs Y_{\gamma} + \su^3 \bs N$, it follows that  \eqref{eq:linfieldequations} becomes 
\begin{align}
	\begin{split}
		&\Div \bs \sigma + \bs f = \bs 0, \quad \mbox{on } \cl B, \\
		&\bs \sigma \bs N = \mu_s \tensor{\msu}{_{,\al}^\al} + (\lambda_s +  \mu_s) \tensor{\su}{^\gamma_{,\al \gamma}} \bs Y^{,\al} - (\zeta + 2 \eta) \p_{\al} \p_\beta\tensor{\bs u}{^{,\al \beta}} + \bs t, \quad \mbox{on } \cl S, \\
		&\bs u = \bs 0, \quad \mbox{on } \p \cl B \backslash \cl S.
	\end{split}\label{eq:HPlinfieldequations}
\end{align}
{ We remark that if $\cl B$ is bounded with a sufficiently smooth boundary, $\cl S \neq \p \cl B$, $\cl S$ has sufficiently smooth boundary, $\bs f \in L^{6/5}(\cl B)$, and $\bs t \in L^p(\cl S)$ with $p > 1$, then \eqref{eq:HPlinfieldequations} has a unique weak solution in an appropriately defined energy space (see Theorem 2 in \cite{Eremeyevetal21}).} 

\section{Mode-III Fracture Problem}

In this section, we apply the linearized theory \eqref{eq:HPlinfieldequations} to the problem of a brittle infinite plate, with a straight crack $\cl C$ of length $2 \ell$, under far-field anti-plane shear loading $\sigma$. As discussed in Section 1 and in contrast to ascribing either a quadratic Gurtin-Murdoch or Steigmann-Ogden surface energy \eqref{eq:classen} to the crack fronts, the use of \eqref{eq:surfaceen} yields a model that predicts bounded strains and stresses up to the crack tips (see Theorem \ref{t:bounded}). 

\subsection{Formulation and governing equations}

We consider a brittle, infinite plate under anti-plane shear loading, $\lim_{x^2 \rar \pm \infty} \sigma_{13} = 0$ and $\lim_{x^2 \rar \pm \infty}\sigma_{23} = \sigma$, with a straight crack $\cl C = \{(x^1, 0, x^3) \mid x^1 \in [-\ell, \ell] \}$ of length $2 \ell$ (see Figure \ref{f:2}). For anti-plane shear, the displacement field takes the form
\begin{align}
	\bs u(x^1, x^2, x^3) = u(x^1, x^2) \bs e_3,
\end{align} 
Then the only nonzero components of the stress are 
\begin{align}
	\sigma_{13} = \mu u_{,1}, \quad \sigma_{23} = \mu u_{,2}. \label{eq:sigma}
\end{align}
By the symmetry of the problem, $u$ can be taken to be even in $x^1$ and odd in $x^2$, so we will focus only on the strain and stress fields for $x^2 \geq 0$. The governing field equations are  \eqref{eq:HPlinfieldequations} on $\cl B = \{ (x^1,x^2,x^3) \mid x^2 \geq 0\}$ with $\cl S = \{ (x^1,0,x^3) \mid x^1 \in [-\ell,\ell] \}$, $\bs t = \bs 0$ and $\bs f = \bs 0$. 

\begin{figure}[b]
	\centering
	\includegraphics[scale=.75]{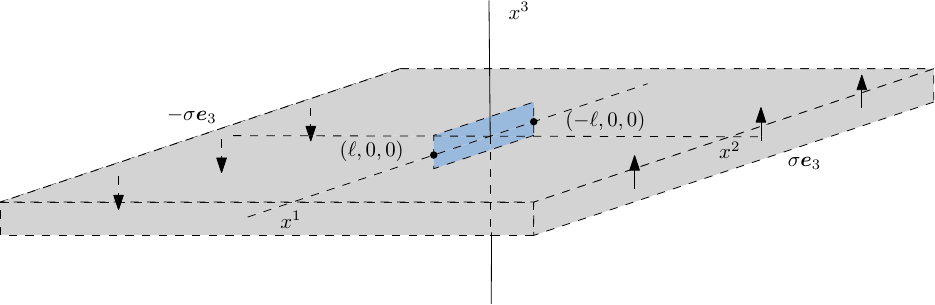}
	\caption{Schematic of the mode-III problem with the crack $\cl C$ appearing in blue.}
	\label{f:2}
\end{figure}

We define dimensionless variables
\begin{align}
x = \frac{x^1}{\ell}, \quad y = \frac{x^2}{\ell}, \quad z = \frac{x^3}{\ell}, \quad
w(x,y,z) = \frac{1}{\ell} \Bigl (
u(x^1, x^2, x^3) - \frac{\sigma}{\mu} x^2
\Bigr ). \label{eq:dimensionless}
\end{align}
Then the field equations take the dimensionless form 
\begin{align}
	\begin{split}
	&\Delta w(x,y) = 0, \quad y > 0, \\
	&-w_y(x,0) = \al w_{xx}(x,0) - \beta w_{xxxx}(x,0) + \gamma, \quad x \in (-1,1),\\
	&w(x,0) = 0, \quad |x| \geq 1, \\
	&w_x(\pm 1, 0) = 0,
	\end{split}\label{eq:field}
\end{align} 
with the decay condition $\lim_{y \rar \infty} |\nabla w(x,y)| = 0$. We note that the boundary conditions $w_x(\pm 1, 0) = 0$ imply that the crack opening is cusp shaped rather then blunted (see also Figure \ref{f:3} and Figure \ref{f:4}). The dimensionless parameters $\al, \beta$ and $\gamma$ are given by  
\begin{align}
	\al = \frac{\mu_s}{\mu \ell} > 0, \quad \beta = \frac{\zeta + 2 \eta}{\mu \ell^3} > 0, \quad \gamma = \frac{\sigma}{\mu}, \label{eq:size}
\end{align}
and in particular, we see from \eqref{eq:size} that the behavior of the displacement $w$ depends on the length of the crack, $\ell$. For macro cracks satisfying $\beta \ll \alpha \ll 1$, we expect $w(x,0)$ to be well-approximated by the singular, rounded opening profile from the classical linear elastic fracture mechanics except in small regions near the crack tips (boundary layers). See Figure \ref{f:4}.  

We remark that in using the Steigmann-Ogden surface energy \eqref{eq:classen} rather than \eqref{eq:surfaceen}, the boundary conditions at $y = 0$ are replaced by 
\begin{align}
	\begin{split}
	&-w_y(x,0) = \al w_{xx}(x,0) + \gamma, \quad x \in (-1,1), \\
	&w(x,0) = 0, \quad |x| \geq 1.
	\end{split}\label{eq:gurtin}
\end{align}
One may view this loss of higher order derivatives in the boundary conditions as a consequence of the fact that the Steigmann-Ogden surface energy does not satisfy the strong-ellipticity condition \eqref{eq:strongellipt}: for anti-plane shear, $\bs b = \bs u(x^1,0,x^3)$ is orthogonal to the surface's normal $\bs n = -\bs e_2$ (see \eqref{eq:LHcondition}).  As discussed in \cite{WaltonNote12, IMRUSch13}, the boundary conditions \eqref{eq:gurtin} \emph{do not} lead to a model predicting bounded strains up to the crack tips $x = \pm 1$, i.e., the displacement field satisfies
\begin{align}
	\sup_{y > 0} |\nabla w(\pm 1, y)| = \infty. 
\end{align}

We see that \eqref{eq:field} is the system of Euler-Lagrange equations for the energy functional 
\begin{align}
	\cl A_L[w] &= \frac{1}{2} \int_0^\infty \int_{-\infty}^\infty |\nabla w(x,y)|^2 dx dy + \int_{-\infty}^\infty \Bigl ( \frac{\al}{2} |w_x(x,0)|^2 + \frac{\beta}{2} |w_{xx}(x,0)|^2 \Bigr ) dx \\
	&\quad- \gamma \int_{-\infty}^\infty w(x,0) dx
\end{align}
defined for $w$ with $\grad w \in L^2(\{y > 0\})$, $w(\cdot,0) \in H^2(\bbR)$ and $w(x,0) = 0$ for all $|x| \geq 1$. Motivated by this observation, we define the Hilbert space $H$ to be the completion of $C^\infty_c((-1,1))$ under the norm
\begin{align}
	\| f \|^2_H := \int_{-\infty}^\infty \bigl ( \alpha |f'(x)|^2 + \beta |f''(x)|^2 \bigr ) dx. \label{eq:defHnorm}
\end{align}
It is straightforward to verify the following facts using the fundamental theorem of calculus and Cauchy-Schwarz inequality:
\begin{itemize}
	\item (Sobolev embedding) If $f \in H$ then $f \in C_c^{1,1/2-}(\bbR)$ and $f(x) = 0$ for all $|x| \geq 1$, and for all $\delta \in [0,1/2)$, there exists a constant $A > 0$ depending on $\al, \beta$ and $\delta$, such that for all $f \in H$,
	\begin{align}
		\| f \|_{C^{1,\gamma}(\bbR)} \leq A \| f \|_H. \label{eq:Sobembedding}
	\end{align}
	\item  $f \in H$ if and only if $f \in H^2(\bbR)$ and $f(x) = 0$ for all $|x| \geq 1$. Moreover, there exist $b, B > 0$ depending on $\al$ and $\beta$ such that for all $f \in H$, 
	\begin{align}
		b \| f \|_H \leq \| f \|_{H^2(\bbR)} \leq B \| f \|_H. \label{eq:equivnorms} 
	\end{align}
\end{itemize}

The problem \eqref{eq:field} can be reduced completely to a problem on the boundary by using the Dirichlet-to-Neumann map $-w_y(x,0) = \cl H w_x(x,0)$ where $\cl H$ is the Hilbert transform
\begin{align}
	\cl H f(x) = \frac{1}{\pi} \, \mbox{p.v.} \, \int_{-\infty}^\infty \frac{f(s)}{x-s} ds, \quad f \in H. 
\end{align}
Then finding $w$ with $\nabla w \in L^2(\{y > 0\})$ and $w(\cdot,0) \in H$ satisfying \eqref{eq:field} is equivalent to determining $w(\cdot,0) =: f \in H$
satisfying\footnote{Once $f$ is found, $w$ is determined on the upper half plane using the standard Poisson kernel for the upper half plane.}
\begin{align}
	\beta f''''(x) - \alpha f''(x) + \cl H f'(x)= \gamma , \quad x \in (-1,1). \label{eq:integrodiff} 
\end{align}
By using the Plancherel theorem, the Fourier representation of the Hilbert transform (see \eqref{eq:Hsymbol}), and \eqref{eq:equivnorms}, we have for all $f \in H$,
\begin{align}
	\| \cl H f' \|_{H^1(\bbR)} \leq \| f' \|_{H^1(\bbR)} \leq \| f \|_{H^2(\bbR)} \leq B \| f \|_H. \label{eq:Hilberth1}
\end{align}

\begin{defn}
A function $f \in H$ is a \emph{weak solution} to the integro-differential equation \eqref{eq:integrodiff} if for all $g \in H$
\begin{align}
	\int_{-\infty}^\infty [\beta f''(x) g''(x) + \alpha f'(x) g'(x) +  \cl H f'(x) g(x) \Bigr ] dx
	= \int_{-\infty}^\infty \gamma g(x) dx. \label{eq:weak}
\end{align}
We remark that since $f, g \in H$, the integrals appearing in \eqref{eq:weak} are in fact over the interval $(-1,1)$.

A function $f \in H$ is a \emph{classical solution} to \eqref{eq:integrodiff} if $f \in C^4((-1,1)) \cap H$ and $f$ satisfies \eqref{eq:integrodiff} pointwise.   
\end{defn}

We note that by \eqref{eq:Hilberth1} and Cauchy-Schwarz, \eqref{eq:weak} is well-defined for each $f,g \in H$.  

\subsection{Solution of the integro-differential equation}

We now establish that there exists a unique classical solution to \eqref{eq:integrodiff}, and the solution's behavior is consistent with the linearization assumption \eqref{eq:infintesimaldisplacement}. 
We denote the following Green function 
\begin{align}
	G(x,\tau) = 
	\begin{cases}
		\frac{1}{24}(x-1)^2(\tau+1)^2(1 + 2x - 2\tau - x\tau) &\quad \tau \in [-1,x], \\
		\frac{1}{24}(\tau-1)^2(x+1)^2(1 + 2\tau - 2x - x\tau) &\quad \tau \in [x,1], 
	\end{cases} 
\end{align}
satisfying $G_{xxxx}(x,\tau) = \delta(x-\tau)$, $G(\pm 1, \tau) = 0$, $G_x(\pm 1, \tau) = 0$. We note that $G(x,\tau) = G(\tau,x)$ for all $\tau,x \in [-1,1]$, $G \in C^2([-1,1] \times [-1,1])$ and $\int_{-1}^1 G(x,\tau) d\tau = \frac{1}{24}(1 - x^2)^2$. In particular, we have for all $f \in H$, 
\begin{align}
	\int_{-1}^1 G_{\tau \tau}(x,\tau) f''(\tau) d\tau = f(x). \label{eq:Greenid}
\end{align}

\begin{figure}[t]
	\centering
	\includegraphics[width=\linewidth]{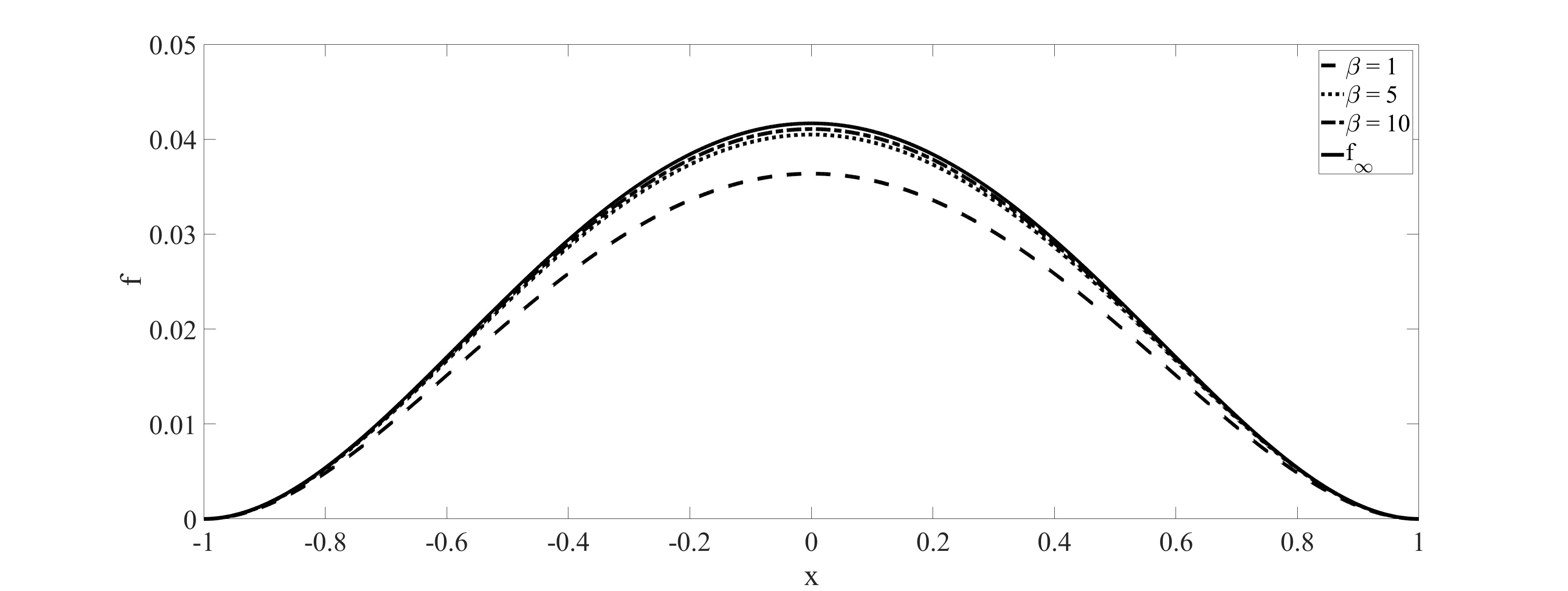}
	\caption{Numerical solutions for the equivalent formulation of \eqref{eq:integrodiff} as a Fredholm problem \eqref{eq:fredequation}. The parameters $(\beta, \alpha, \gamma)$ range over $(1,1,1)$, $(5,1,5)$ and $(10,1,10)$. For $\gamma = \beta$ and $\beta \gg 1 \simeq \alpha$, the opening profile is well approximated by the limiting opening profile $f_{\infty}(x) = \frac{1}{24}(1-x^2)^2$ on $[-1,1]$.} 
	\label{f:3}
\end{figure}

\begin{lem}\label{l:Green}
		A function $f \in H$ is a weak solution to \eqref{eq:integrodiff} if and only if $f$ satisfies
\begin{align}
	\beta f(x) + \int_{-1}^1 G(x,\tau) (-\alpha f''(\tau) + \cl H f'(\tau) ) d\tau = \frac{\gamma}{24}(1 - x^2)^2, \quad x \in [-1,1]. \label{eq:integralequation}
\end{align}
\end{lem}

\begin{proof}
Let $h \in L^2(\bbR)$ with $h(x) = 0$ for all $|x| > 1$, and set 
\begin{align}
	g(x) = 
\begin{cases}
	\int_{-1}^1 G(x,\tau) h(\tau) d\tau &\mbox{ if } |x| \leq 1, \\
	0 &\mbox{ if } |x| > 1. 
\end{cases}
\end{align}
Since $G \in C^2([-1,1] \times [-,1,1])$, $G(\pm 1, \tau) = 0$, and $G_x(\pm 1, \tau) = 0$, $g$ is twice continuously differentiable on $\bbR \backslash \{\pm 1\}$ and continuously differentiable on $\bbR$ with 
\begin{align}
	g'(x) &= \chi_{\{|x| \leq 1\}}(x) \int_{-1}^1 G_x(x,\tau) h(\tau) d\tau, \\
	g''(x) &= \chi_{\{|x| \leq 1\}}(x) \int_{-1}^1 G_{xx}(x,\tau) h(\tau) d\tau, 
\end{align}  
where $\chi_E$ is the indicator function of a subset $E \subseteq \bbR$.
In particular, we conclude that $g \in H^2(\bbR)$ and thus, $g \in H$. Inserting $g$ into \eqref{eq:weak}, integrating by parts in the second term and using that $g(\pm 1) = 0$ yield
\begin{gather}
	\beta \int_{-1}^1 \int_{-1}^1 G_{xx}(x,\tau)f''(x) h(\tau) d\tau dx \\ + \int_{-1}^1 \int_{-1}^1 G(x,\tau) (-\alpha f''(x) + \cl H f'(x)) h(\tau) d\tau dx \\
	= \gamma \int_{-1}^1 \int_{-1}^1 G(x,\tau) h(\tau) d\tau dx. 
\end{gather}
Interchanging the order of integration, using the symmetry of $G$ and relabeling the integration variables lead to 
\begin{gather}
		\beta \int_{-1}^1 \int_{-1}^1 G_{\tau \tau}(x,\tau)f''(\tau) d\tau \, h(x) dx \\ + \int_{-1}^1 \int_{-1}^1 G(x,\tau) (-\alpha f''(\tau) + \cl H f'(\tau)) d\tau \, h(x) dx \\
	= \int_{-1}^1 \frac{\gamma}{24}(1 - x^2)^2 h(x) dx.
\end{gather} 
Finally, by \eqref{eq:Greenid} we conclude that  
\begin{gather}
\int_{-1}^1 \Bigl [ \beta f(x) + \int_{-1}^1 G(x,\tau) (-\alpha f''(\tau) + \cl H f'(\tau)) d\tau \Bigr ]\, h(x) dx \\
= \int_{-1}^1 \frac{\gamma}{24}(1 - x^2)^2 h(x) dx,
\end{gather}
for all $h(x) \in L^2(\bbR)$ with $h(x) = 0$ for all $|x| > 1$, proving \eqref{eq:integralequation}.

Conversely, if $f \in H$ and \eqref{eq:integralequation} holds, then for all $g \in H$, we have 
\begin{align}
\beta \int_{-\infty}^\infty f''(x) g''(x) dx &= 
 \int_{-1}^1 \int_{-1}^1 G_{xx}(x,\tau) (\alpha f''(\tau) - \cl H f'(\tau)) g''(x) dx \\
&\quad + \int_{-1}^1 \frac{\gamma}{6}(3 x^2 - 1) g''(x) dx.
\end{align}  
We again interchange the order of integration and use integration by parts and $\int_{-1}^1 G_{xx}(x,\tau) g''(x) dx = g(\tau)$ to conclude that
\begin{align}
	\int_{-1}^1 \int_{-1}^1& G_{xx}(x,\tau) (\alpha f''(\tau) - \cl H f'(\tau)) g''(x) dx
	 + \int_{-1}^1 \frac{\gamma}{6}(3x^2 - 1) g''(x) dx \\
	 &= \int_{-1}^1 (\alpha f''(\tau) - \cl H f'(\tau)) g(\tau) d\tau + \int_{-1}^1 \gamma g(x) dx \\
	 &= -\int_{-1}^1 (\alpha f'(x) g'(x) + \cl H f'(x) g(x)) dx + \int_{-1}^1 \gamma g(x) dx
\end{align}
This proves $f \in H$ satisfies \eqref{eq:weak} and concludes the proof of the lemma.
\end{proof}

\begin{figure}[t]
	\centering
	
	\includegraphics[width=1\linewidth]{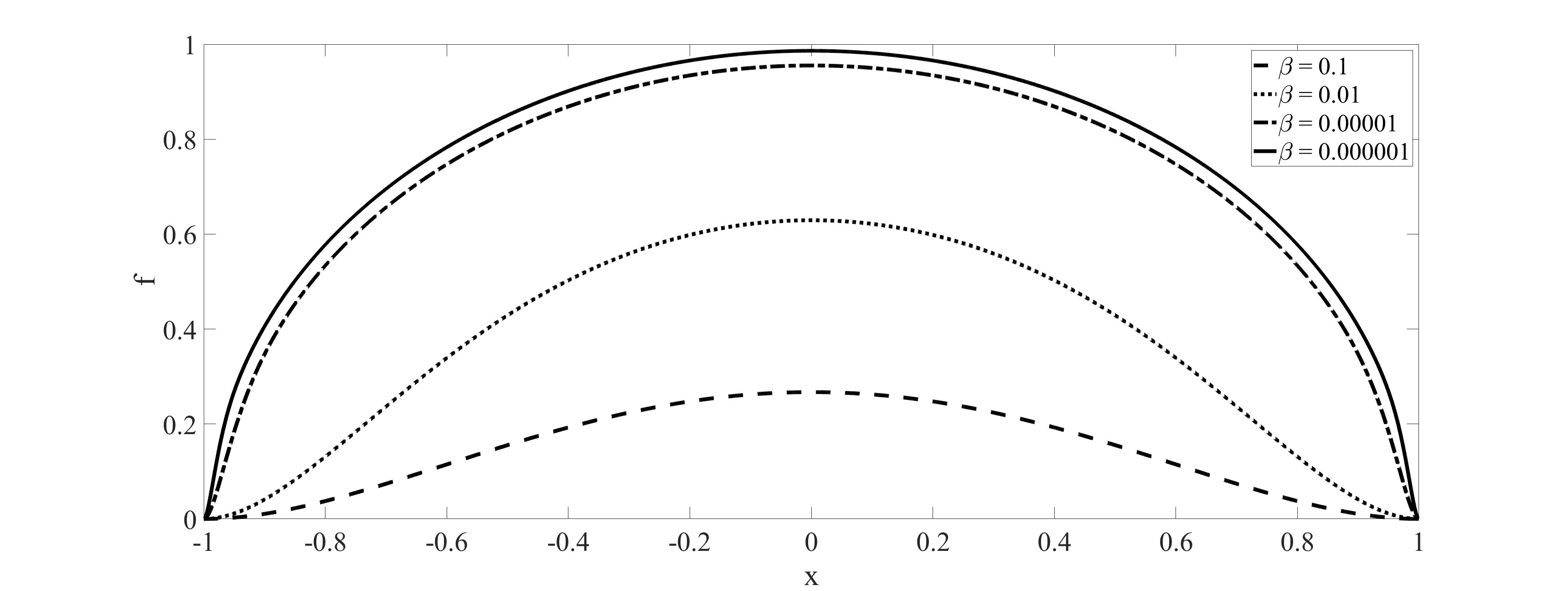}
	\includegraphics[width=1\linewidth]{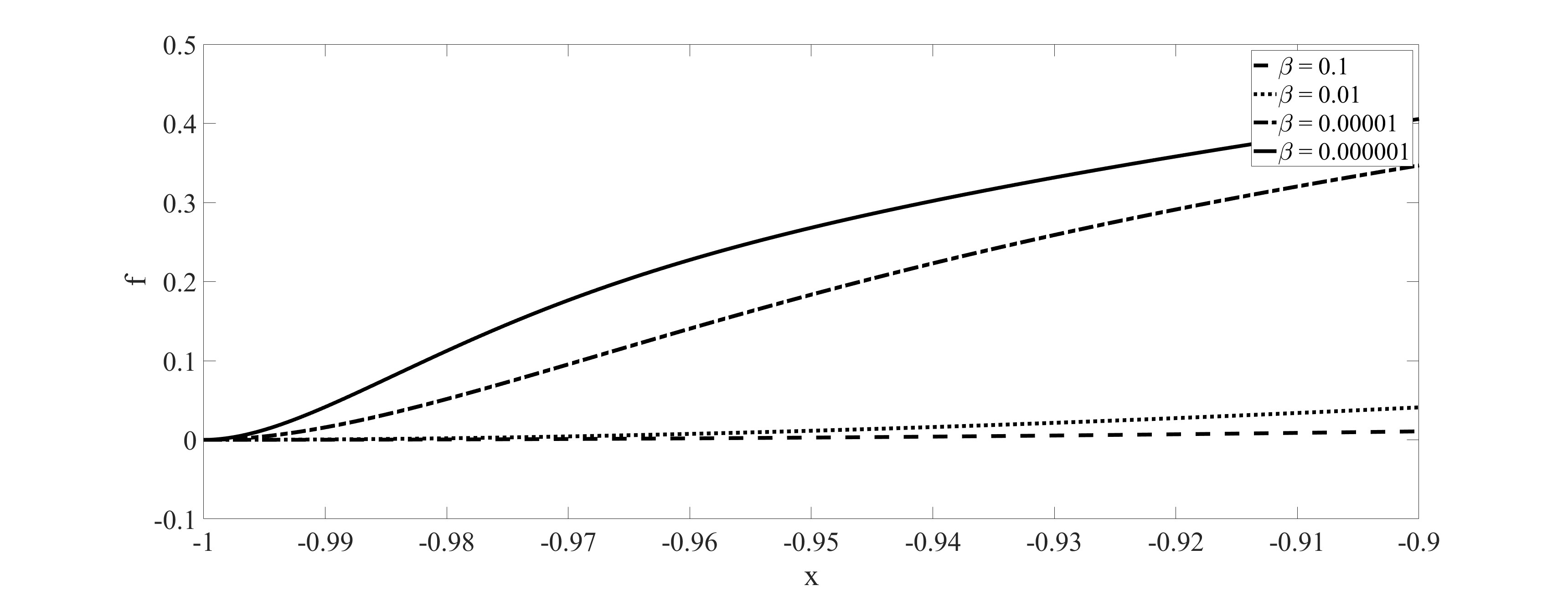}
	
	\caption{Numerical solutions for the macro-crack regime $\beta \ll \al \ll 1$. The parameters $(\beta, \alpha, \gamma)$ range over $(10^{-1},10^{-1},1)$, $(10^{-2},10^{-1},1)$, $(10^{-5},10^{-2},1)$ and $(10^{-6},10^{-3},1)$. For $\beta \ll \alpha \ll 1$, we expect the crack opening $f(x)$ to be well-approximated by the singular, rounded opening profile predicted by classical linear elastic fracture mechanics away from the crack tips where $f'(\pm 1) = 0$ (and the profile is cusped).}
	\label{f:4}
\end{figure}

Via integration by parts and straightforward computations, we conclude from Lemma \ref{l:Green} that the crack opening profile $f$ must satisfy the following Fredholm integral equation of the second kind \eqref{eq:fredequation}. We will show in Theorem \ref{t:bounded} that this equation is uniquely solvable for \emph{arbitrary} $\al, \beta > 0$ and $\gamma$. We remark that since the kernel extends to a continuous function on $[-1,1] \times [-1,1]$ (the singularities are removable), the numerical computation of solutions is relatively straightforward via the N\"ystrom method with the trapezoidal rule to approximate the integral (see Figure \ref{f:3} and Figure \ref{f:4}).  
\begin{cor}\label{c:Fredholm}
	A function $f \in H$ is a weak solution to \eqref{eq:integrodiff} if and only if $f$ satisfies the Fredholm equation
	\begin{align}
		\beta f(x) + \int_{-1}^1 K(x,s) f(s) ds = \frac{\gamma}{24}(1 - x^2)^2, \quad x \in [-1,1], \label{eq:fredequation}
	\end{align}
where $K(x,s) = -\alpha G_{ss}(x,s) + \frac{1}{\pi} \mbox{p.v.}\, \int_{-1}^1 \frac{G_{\tau}(x,\tau)}{s - \tau} d\tau$,
\begin{align}
	G_{ss}(x,s) &= 
\begin{cases}
	-\frac{1}{4}(x-1)^2(2s + xs + 1) &\quad s \in [-1,x], \\
	-\frac{1}{4}(x+1)^2(-2s + xs + 1) &\quad s \in [x,1],
\end{cases}
\end{align}
and 
\begin{align}
\frac{1}{\pi} \mbox{p.v.}\, \int_{-1}^1 \frac{G_{\tau}(x,\tau)}{s - \tau} d\tau &= \frac{1}{4\pi}(sx - 1)(x^2 - 1) + \frac{1}{2\pi}(s - x)^2 \log |x-s| \\
&\quad - \frac{1}{8\pi}(x-1)^2(-x + 2s + sx)(1+s)\log(1+s) \\
&\quad - \frac{1}{8\pi}(x+1)^2(x - 2s + sx)(1-s)\log(1-s). 
\end{align}
\end{cor}

\begin{thm}\label{t:bounded}
	There exists $C > 0$ depending on $\al$ and $\beta$ such that the following hold. There exists a unique classical solution $f$ to \eqref{eq:integrodiff}, and $f$ satisfies 
	\begin{align}
		\| f \|_{C^4([-1,1])} \leq C |\gamma|. \label{eq:festimate}
	\end{align}
	Moreover, the displacement field $w(x,y) = \int_{-\infty}^\infty P_y(x-s) f(s) ds$, where $P_y(\cdot)$ is the Poisson kernel for the upper half plane, has bounded { stresses and} strains up to the crack tips: 
	\begin{align}
	{ \mu^{-1} \| \bs \sigma \|_{C(\{y \geq 0\})}} + \| w \|_{C^1(\{ y \geq 0 \})} \leq C |\gamma|. \label{eq:westimate}
	\end{align} 
\end{thm}

\begin{proof}
In what follows, $C$ will denote a positive constant depending only on $\alpha$ and $\beta$ that may change from line to line, and we denote the Fourier transform and inverse Fourier transform by 
\begin{align}
	\hat f(\xi) = \int_{-\infty}^\infty f(x) e^{-2\pi ix \xi} dx, \quad \check f(x) = \int_{-\infty}^\infty f(\xi) e^{2\pi i x \xi}\, d\xi.
\end{align}
We recall that
\begin{align}
	\widehat{f'}(\xi) = 2 \pi i \xi \hat f(\xi), \quad 
	\widehat{\cl H f}(\xi) = -i \sgn(\xi) \hat f(\xi), \label{eq:Hsymbol}
\end{align}
the latter relation following from the Fourier representation of $w$ on the upper half plane (see \eqref{eq:fourier}). 

We define a bilinear form $B(\cdot,\cdot) : H \times H \rar \bbR$ by 
\begin{align}
	B(f,g) = 	\int_{-\infty}^\infty [\beta f''(x) g''(x) + \alpha f'(x) g'(x) +  \cl H f'(x) g(x) \Bigr ] dx, \quad f,g \in H. 
\end{align}
By Cauchy-Schwarz, \eqref{eq:Hilberth1}, and \eqref{eq:equivnorms} we conclude that for all $f, g \in H$, 
\begin{align}
	|B(f,g)| \leq (1 + B^2) \| f \|_H \| g \|_H,
\end{align}
so that $B(\cdot,\cdot)$ is a bounded bilinear form on $H$. Moreover, by the Plancherel theorem and \eqref{eq:Hsymbol}, we have for all $f \in H$,
\begin{align}
	\int_{-\infty}^\infty \cl H f'(x) f(x) dx = \int_{-\infty}^\infty 2 \pi |\xi| |\hat f(\xi)|^2 d\xi \geq 0. 
\end{align}
Thus, the bilinear form is coercive. Since for all $g \in H$, 
\begin{align}
	\Bigl | \int_{-\infty}^\infty \gamma g(x) dx \Bigr | \leq |\gamma| \| g \|_{C([-1,1])} \leq |\gamma| A \| g \|_H, 
\end{align}
the classical Lax-Milgram theorem implies that there exists a unique weak solution $f \in H$ to \eqref{eq:integrodiff}, and moreover, 
\begin{align}
	\| f \|_H \leq A |\gamma|. \label{eq:lmestimate}
\end{align}

To prove \eqref{eq:westimate}, we express $w$ via the Fourier transform, 
\begin{align}
	w(x,y) = \int_{-\infty}^\infty e^{-2\pi y |\xi|} e^{2\pi i x\xi} \hat f(\xi) d\xi. \label{eq:fourier}
\end{align}
Since $f \in H \subset H^2(\bbR)$, we have 
\begin{align}
	\int (1 + |\xi|)^4 |\hat f(\xi)|^2 d \xi \leq C_0 \| f \|_{H^2(\bbR)}^2 \leq C \| f \|_H^2. 
\end{align}
Thus, by Cauchy-Schwarz 
\begin{align}
	|w(x,y)&| + |\nabla w(x,y)| \\ &\leq \int_{-\infty}^\infty (1 + 2\pi |\xi| \sqrt{2}) |\hat f(\xi)|^2 \, d\xi \\
	&\leq 2\pi \sqrt{2} \Bigl (
	\int_{-\infty}^\infty |\xi|^2(1 + |\xi|)^{-4} d\xi 
	\Bigr )^{1/2} \Bigl (
	\int_{-\infty}^\infty (1 + |\xi|)^{4} |\hat f(\xi)|^2 d\xi 
	\Bigr )^{1/2} \\
	&\leq C \| f \|_H \leq C |\gamma|. \label{eq:wbound}
\end{align}
{ The bound on $\mu^{-1}\bs \sigma$ then follows immediately from \eqref{eq:wbound}, \eqref{eq:dimensionless} and \eqref{eq:sigma}.}

We now show that the weak solution $f$ is a classical solution, and \eqref{eq:festimate} holds. By a density argument in $H$, we have that $\int_{-1}^1 G(\cdot,\tau)(\alpha f''(\tau) - \cl H f'(\tau)) d\tau \in H^3([-1,1])$ with 
\begin{align}
	\frac{d^k}{dx^k} \int_{-1}^1 & G(x,\tau)(\alpha f''(\tau) - \cl H f'(\tau)) d\tau \\ &= 
	 \int_{-1}^1 \p_x^k G(x,\tau)(\alpha f''(\tau) - \cl H f'(\tau)) d\tau, \quad k = 1, 2, \\ 
	\frac{d^3}{dx^3} \int_{-1}^1 &G(x,\tau)(\alpha f''(\tau) - \cl H f'(\tau)) d\tau \\ &= 
	\int_{-1}^x \frac{1}{4}(2-\tau)(\tau+1)^2 (\alpha f''(\tau) - \cl H f'(\tau)) d\tau \\&\quad - 
	\int_{x}^1 \frac{1}{4}(2+\tau)(\tau-1)^2 (\alpha f''(\tau) - \cl H f'(\tau)) d\tau. \label{eq:threederivatives}
\end{align}
By Lemma \ref{l:Green}, we conclude that $f \in H^3([-1,1])$, and by \eqref{eq:integralequation} \eqref{eq:threederivatives}, Cauchy-Schwarz, and \eqref{eq:lmestimate} we have
\begin{align}
\| f''' \|_{L^2([-1,1])} 
&\leq C ( \| f'' \|_{L^2(\bbR)} + \| \cl H f' \|_{L^2(\bbR)} + |\gamma| ) \\ 
&\leq C (\| f \|_H + |\gamma|) \leq C |\gamma|. \label{eq:threederest}
\end{align}  
Moreover, by \eqref{eq:lmestimate}, \eqref{eq:threederest} and the fundamental theorem of calculus, $f \in C^2([-1,1])$ with 
\begin{align}
	\| f \|_{C^2([-1,1])} \leq C |\gamma|. \label{eq:twoderivatives}
\end{align}

By \eqref{eq:weak} and integration by parts, it follows that for all $g \in C^\infty_c((-1,1)) \subset H$, 
\begin{align}
	 \int_{-1}^1 f'''(x) g'(x) dx = \frac{1}{\beta} \int_{-1}^1 [ -\alpha f''(x) + \cl H f'(x) - \gamma ] g(x) dx,
\end{align}
and, thus, $f \in H^4([-1,1])$ and
\begin{align}
		\beta f''''(x) - \alpha f''(x) + \cl H f'(x)= \gamma , \quad \mbox{for a.e. }x \in (-1,1). \label{eq:aeint}
\end{align}
Then by \eqref{eq:aeint}, \eqref{eq:twoderivatives}, and the fact that $\cl H f' \in H^1(\bbR) \hookrightarrow C(\bbR)$, we conclude that $f'''' \in C([-1,1])$ and 
\begin{align}
	\| f'''' \|_{C([-1,1])} &\leq C \Bigl ( \| f \|_{C^2([-1,1])} + \| \cl H f' \|_{C([-1,1])} + |\gamma| \Bigr ) \\
	&\leq C \Bigl ( \| \cl H f' \|_{H^1(\bbR)} + |\gamma| \Bigr ) \\
	&\leq C \Bigl ( \| f \|_H + |\gamma| \Bigr ) \leq C |\gamma|. \label{eq:fourderest}
\end{align}
By \eqref{eq:threederest}, \eqref{eq:twoderivatives} and \eqref{eq:fourderest}, it follows that $f \in C^4([-1,1])$ is a classical solution to \eqref{eq:integrodiff} and \eqref{eq:festimate} holds.
\end{proof}

\bibliographystyle{plain}
\bibliography{researchbibmech}
\bigskip

\centerline{\scshape C. Rodriguez}
\smallskip
{\footnotesize
	\centerline{Department of Mathematics, University of North Carolina}
	
	\centerline{Chapel Hill, NC 27599, USA}
	
	\centerline{\email{crodrig@email.unc.edu}}
}

\end{document}